\documentclass[a4paper,USenglish]{lipics-v2016}

\usepackage{graphicx}

\usepackage{wrapfig}
\usepackage{url}

\usepackage{amssymb}

\usepackage{algorithm}
\usepackage[noend]{algpseudocode}

\usepackage{etoolbox}
\usepackage{tikz}

\usetikzlibrary{tikzmark}
\usetikzlibrary{calc}

\errorcontextlines\maxdimen

\newcommand{\ALGtikzmarkcolor}{black}
\newcommand{\ALGtikzmarkextraindent}{4pt}
\newcommand{\ALGtikzmarkverticaloffsetstart}{-.5ex}
\newcommand{\ALGtikzmarkverticaloffsetend}{-.5ex}
\makeatletter
\newcounter{ALG@tikzmark@tempcnta}

\newcommand\ALG@tikzmark@start{%
    \global\let\ALG@tikzmark@last\ALG@tikzmark@starttext%
    \expandafter\edef\csname ALG@tikzmark@\theALG@nested\endcsname{\theALG@tikzmark@tempcnta}%
    \tikzmark{ALG@tikzmark@start@\csname ALG@tikzmark@\theALG@nested\endcsname}%
    \addtocounter{ALG@tikzmark@tempcnta}{1}%
}

\def\ALG@tikzmark@starttext{start}
\newcommand\ALG@tikzmark@end{%
    \ifx\ALG@tikzmark@last\ALG@tikzmark@starttext
    \else
        \tikzmark{ALG@tikzmark@end@\csname ALG@tikzmark@\theALG@nested\endcsname}%
        \tikz[overlay,remember picture] \draw[\ALGtikzmarkcolor] let \p{S}=($(pic cs:ALG@tikzmark@start@\csname ALG@tikzmark@\theALG@nested\endcsname)+(\ALGtikzmarkextraindent,\ALGtikzmarkverticaloffsetstart)$), \p{E}=($(pic cs:ALG@tikzmark@end@\csname ALG@tikzmark@\theALG@nested\endcsname)+(\ALGtikzmarkextraindent,\ALGtikzmarkverticaloffsetend)$) in (\x{S},\y{S})--(\x{S},\y{E});%
    \fi
    \gdef\ALG@tikzmark@last{end}%
}

\apptocmd{\ALG@beginblock}{\ALG@tikzmark@start}{}{\errmessage{failed to patch}}
\pretocmd{\ALG@endblock}{\ALG@tikzmark@end}{}{\errmessage{failed to patch}}
\makeatother

\usepackage{caption}

\usepackage{thmtools}
\usepackage{thm-restate}

\usepackage{hyperref}

\usepackage{tabularx}

\usepackage[capitalise,noabbrev]{cleveref}

{\bfseries}{\itshape}
{\bfseries}{\itshape}
\crefname{fact}{Fact}{Fact}
\crefname{lem}{Lemma}{Lemma}

\algnewcommand\algorithmicforeach{\textbf{for each}}
\algdef{S}[FOR]{ForEach}[1]{\algorithmicforeach\ #1\ \algorithmicdo}

\newcounter{PotentialIndex}

\newcommand{\HPT}{\text{HPT}}
\newcommand{\DHT}{\text{DHT}}
\newcommand{\PTP}{\text{P2P}}
\newcommand{\concat}{\circ}
\newcommand{\MsdSet}{\text{MSD}}

\newcommand{\PatSet}{\text{PAT}}
\newcommand{\TrieSet}{\text{HPT}}

\newcommand{\KeySet}{\text{KEYS}}
\newcommand{\MsgSet}{\text{MSG}}

\newcommand{\bitsSym}{d}
\newcommand{\MemoryBound}{$\Theta (\bitsSym)$}

\newcommand{\PN}{\text{SHPT}}

\makeatletter
\newcommand*{\rom}[1]{\expandafter\@slowromancap\romannumeral #1@}
\makeatother

\title{A Self-Stabilizing Hashed Patricia Trie\footnote{This work was partially supported by the German Research Foundation (DFG) within the Collaborative
Research Center 'On-The-Fly Computing' (SFB 901).}}

\author[1]{Till Knollmann}
\author[2]{Christian Scheideler}
\affil[1]{Heinz Nixdorf Institute \& Computer Science Department,\\Paderborn University, Germany\\
\href{mailto:tillk@mail.upb.de}{tillk@mail.upb.de}\\
\href{https://www.hni.uni-paderborn.de/alg/}{https://www.hni.uni-paderborn.de/alg/}}
\affil[2]{Computer Science Department, Paderborn University, Paderborn, Germany\\
\href{mailto:scheideler@upb.de}{scheideler@upb.de}\\
\href{https://cs.uni-paderborn.de/ti/}{https://cs.uni-paderborn.de/ti/}}
\authorrunning{T.\ Knollmann and C.\ Scheideler} 
 
\keywords{Self-Stabilizing, Prefix Search, Distributed Data Structure}

\begin{document}

\maketitle

\begin{abstract}
	While a lot of research in distributed computing has covered solutions for self-stabilizing computing and topologies, there is far less work on self-stabilization for distributed data structures.
	Considering crashing peers in peer-to-peer networks, it should not be taken for granted that a distributed data structure remains intact.
	In this work, we present a self-stabilizing protocol for a distributed data structure called the \emph{hashed Patricia Trie} (Kniesburges and Scheideler WALCOM'11) that enables efficient prefix search on a set of keys.
	The data structure has a wide area of applications including string matching problems while offering low overhead and efficient operations when embedded on top of a distributed hash table.
	Especially, longest prefix matching for $x$ can be done in $\mathcal{O}(\log |x|)$ hash table read accesses.
	We show how to maintain the structure in a self-stabilizing way.
	Our protocol assures low overhead in a legal state and a total (asymptotically optimal) memory demand of $\Theta(d)$ bits, where $d$ is the number of bits needed for storing all keys.
\end{abstract}

\section{Introduction}

We consider the problem of maintaining a distributed data structure for efficient \emph{Longest Prefix Matching} in peer-to-peer (P2P) systems.
We focus on the \emph{hashed Patricia Trie} (\HPT) introduced in \cite{HashedPatriciaTrie} and present an algorithm rendering a self-stabilizing version of this data structure when applied on top of any reliable \emph{distributed hash table} (\DHT).

\begin{definition}(Longest Prefix Matching)
	\label{definition:prefix_search}
	Consider a set of binary strings called \emph{keys} and a binary string $x$.
	The task of Longest Prefix Matching is to find a key $y$ sharing the longest common prefix with $x$.
	A prefix of a binary string is a substring beginning with the first bit.
	We denote the longest common prefix of $x$ and $y$ by $\ell cp(x,y)$.
\end{definition}

We denote a prefix $p$ of $x$ by $p\sqsubseteq x$.
$p$ is a \emph{proper prefix} of $x$ ($p\sqsubset x$) if $p$ is a prefix of $x$ and $\vert p\vert < \vert x\vert$, where $|p|$ is the length of $p$.
Longest Prefix Matching is an old problem with applications in various areas including string matching problems and IP lookup in Internet routers.
To solve it efficiently in a distributed P2P system, the \HPT\ has been introduced \cite{HashedPatriciaTrie}.
The \HPT\ is a distributed data structure applied to any common \DHT\ which allows efficient prefix search for $x$ in $\mathcal{O}(\log \vert x\vert)$ read accesses to the hash table, i.e., solely based on the length of the search word $x$.
The costs for an insertion of $x$ is in $\mathcal{O}(\log \vert x\vert)$ read accesses and $\mathcal{O}(1)$ write accesses, while deletion can be done in $\mathcal{O}(1)$ accesses.
The memory space used is asymptotically optimal in $\Theta (\text{sum of all key lengths})$.
Moreover, \emph{Suffix Trees} can be implemented efficiently using Patricia Tries and thus also hashed Patrica Tries (called \emph{PAT Trees} \cite{Gonnet:1992:NIT:129687.129692}).
This allows us to efficiently decide if a given string $x$ is a substring of a text in a runtime only depending on the length of $x$.

The usefulness of Patricia Tries motivates us to investigate how a \HPT\ can be maintained in a P2P system where nodes may enter/leave or even fail.
While a lot of research has considered the design of self-stabilizing computation or topologies (See \cref{subsection:related_work}), to the best of our knowledge there are far fewer results concerning self-stabilizing distributed data structures.
However, considering failures of peers, the stability of any distributed data structure can also be affected.
Therefore, we consider the problem of finding an efficient distributed protocol to rebuild the \HPT\ in a self-stabilizing way.

\subsection{Model}

We assume the existence of a self-stabilizing \emph{distributed hash table} (DHT) which provides the operations $\textsc{DHT-Insert}(x)$ to insert data and $\textsc{DHT-Search}(x)$ to retrieve data.
These operations are carried out reliably on the stored data, i.e., no operation is ever canceled.
We assume the existence of a collision-free hash function which maps binary strings to positions in $[0,1)$ to store data in the \DHT.
The function is available locally at every peer.
Each peer has a unique identifier, manages local variables and maintains a \emph{channel}.
When a peer sends a message $m$ to peer $p$, it puts $m$ in the channel of $p$.
A channel has unbounded capacity and messages never get lost.
If a peer processes a message in its channel, the message is removed from the channel afterwards.

We distinguish between two types of \emph{actions}: The first one is for standard procedures and has the form $\langle label \rangle (\langle parameters \rangle):\langle command \rangle$ where \emph{label} is the name of the action, \emph{parameters} define the set of parameters and \emph{command} defines the statements that are executed when calling the action.
It may be executed locally or remotely.
The second type has the form $\langle label \rangle:(\langle guard\rangle)\rightarrow \langle command\rangle$ where \emph{label} and \emph{command} are defined as above and \emph{guard} is a predicate over local variables.
An action at peer $p$ can only be executed if its guard is \emph{true} or a message in the channel of $p$ requests to call it.
We call such an action \emph{enabled}.
The guard of our protocol routine $\textsc{Timeout}$ is always \emph{true}.

A \emph{state} of the system is defined by the assignment of variables at every peer, the data items and their values stored at every peer and all messages in channels of peers.
The system can transform from a state $s$ to another state $s'$ by execution of an enabled \emph{action} at a peer.
An infinite sequence of states $(s_1,s_2,\dots)$ is a \emph{computation} if $s_{i+1}$ can be reached by execution of an action enabled in $s_{i}$ for all $i\geq 1$.
The state $s_1$ is called \emph{initial state}.
We assume \emph{fair message receipt}, i.e., every message contained in a channel is eventually processed.
Also, we assume \emph{weakly fair action execution} such that any action that is enabled in all but finitely many states is executed infinitely often.
This especially applies to the $\textsc{Timeout}$ procedure.
Most importantly, our protocol is \emph{self-stabilizing}.
We call a protocol self-stabilizing, if it fulfills \emph{Convergence} and \emph{Closure}.
Convergence means that starting from an arbitrary initial state, the protocol transforms the system to a legal state in finite time.
Closure denotes that starting from a legal state, the protocol only transforms the system to consecutive legal states.\enlargethispage{\baselineskip}
Our goal is to reach a state in which the \HPT\ provided by the system is in a \emph{legal state}.
We define the legal state of a \HPT\ later in \cref{section:protocol_analysis:correctness}.

\subsection{Related Work}\label{subsection:related_work}

The basic data structure we consider here is the Patricia Trie.
This compressed tree structure has been introduced by Morrison in \cite{Morrison:1968:PAR:321479.321481}.
It was extended to the hashed Patricia Trie by Kniesburges and Scheideler in \cite{HashedPatriciaTrie}.
In \cite{Gonnet:1992:NIT:129687.129692}, Gonnet et al.\ presented PAT Trees which are essentially Patricia Tries for special suffixes (\emph{sistrings}) of a text.
This widens the applications of Patricia Tries to general string problems such as deciding if a word or sentence is contained in a text \cite{Gonnet:1992:NIT:129687.129692}.
The work on self-stabilization started with the research of Dijkstra in \cite{Dijkstra:1974:SSS:361179.361202} where he analyzed self-stabilization in a token ring scenario.
Since then, research has covered wide areas including self-stabilizing computation \cite{awerbuch1991distributed, collin1994self} and coordination \cite{afek1990memory, arora1994distributed, Dijkstra:1974:SSS:361179.361202, flatebo1994two}.
Furthermore, with the rise of P2P systems \cite{rowstron2001pastry, stoica2001chord}, self-stabilizing topologies in the sense of overlay networks gained attraction \cite{Clouser:2012:TSD:2161009.2161351, cramer2005self, dolev2008hypertree, jacob2014skip+, jacob2009self, Kniesburges2014, shaker2005self}.
We use approaches originally presented for topological self-stabilization.
This includes a technique called \emph{Linearization} presented by Onus et al.\ in \cite{Onus:2007:LLS:2791188.2791198}.
A common approach for storing data in overlay networks is a distributed hash table (\DHT ) like Chord \cite{stoica2001chord}.
Using hashing, data items, as well as network peers, are mapped to the $[0,1)$ interval such that a mapping between them is established.
There are various results on self-stabilizing \DHT s in the literature (for example \cite{Kniesburges2014}).
Further, most (self-stabilizing) overlay networks can easily be extended to a \DHT\ given sortable unique identifiers for the peers which is a common assumption.

\subsection{Our Contribution}

We present a self-stabilizing protocol called \PN\ to maintain a slightly modified version of the \HPT\ as presented in \cite{HashedPatriciaTrie}.
Whenever we refer to \HPT, we implicitly mean the modified version.
The \HPT\ and our modification are briefly introduced in \cref{section:HashedPatriciaTrie}.
Afterwards, \cref{section:SHPT_protocol} gives a high-level description of the most important mechanisms of our protocol.
We only require for an initial state that the underlying \DHT\ is in a legal state and that a set of unique keys is stored at \DHT\ nodes.
In \cref{section:Protocol_analysis}, we show that our protocol stabilizes a \HPT\ in finite time out of any initial state.
When the \HPT\ is in a legal state, our protocol guarantees a low overhead of a constant amount of hash table read accesses and messages generated at each \DHT\ node per call of the protocol routine.
Furthermore, we can bound the total memory consumption in a legal state to \MemoryBound\ bits if $\bitsSym$ is the number of bits needed to store all keys.
To improve readability, we deferred the Pseudocode to \cref{appendix:pseudo-code} and the full proofs concerning correctness and overhead to \cref{appendix:correctness_proof} and \cref{appendix:overhead_proof}.

\section{Hashed Patricia Trie}\label{section:HashedPatriciaTrie}

We consider a data structure called the \emph{hashed Patricia Trie} (\HPT) as presented in \cite{HashedPatriciaTrie}.
The \HPT\ is an extended Patricia Trie that is distributed in a \PTP\ System by using a \DHT.
We briefly describe the construction.
For details, we refer to \cite{HashedPatriciaTrie}.
The Patricia Trie is a compressed trie which was proposed by Morrison in \cite{Morrison:1968:PAR:321479.321481}.
Suppose we are given a key set $\KeySet$ consisting of strings.
A trie is a tree structure that consists of labeled nodes and labeled edges.
The root node is labeled by the empty string and every edge is labeled by one character.
The label of a node is the concatenation of all edge labels of edges traversed on the unique path from the root to the node.
For each $k\in\KeySet$ there is a node labeled by $k$ (see \cref{figure:Trie_Example}).
The Patricia Trie introduces compression by allowing edge labels to be strings such that inner nodes with a single child, which do not represent a key, can be avoided.
Similar to \cite{HashedPatriciaTrie}, we restrict ourselves to keys represented by binary strings.
We store the Patricia Trie in a \DHT\ by hashing all nodes by their label resulting in the hashed Patricia Trie.
Our notation is close to the one of \cite{HashedPatriciaTrie} and can be seen in \cref{figure:TRIE_Node_Values}.

\begin{figure}[htb]
	\begin{minipage}[b][][b]{0.25\textwidth}
		\centering
		\includegraphics[width=0.85\textwidth, trim = 0cm 13cm 30cm 0cm, clip=true]{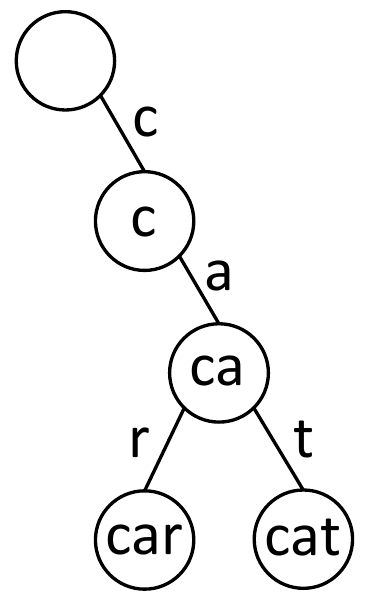}
		\captionof{figure}{Example of a classical Trie containing the keys "car" and "cat".\\}
		\label{figure:Trie_Example}
	\end{minipage}
	\begin{minipage}{0.05\textwidth}
		\hfill
	\end{minipage}
	\begin{minipage}[b][][b]{0.68\textwidth}
		\centering
		\includegraphics[width=0.9\textwidth, trim = 0cm 10.6cm 21.5cm 0cm, clip=true]{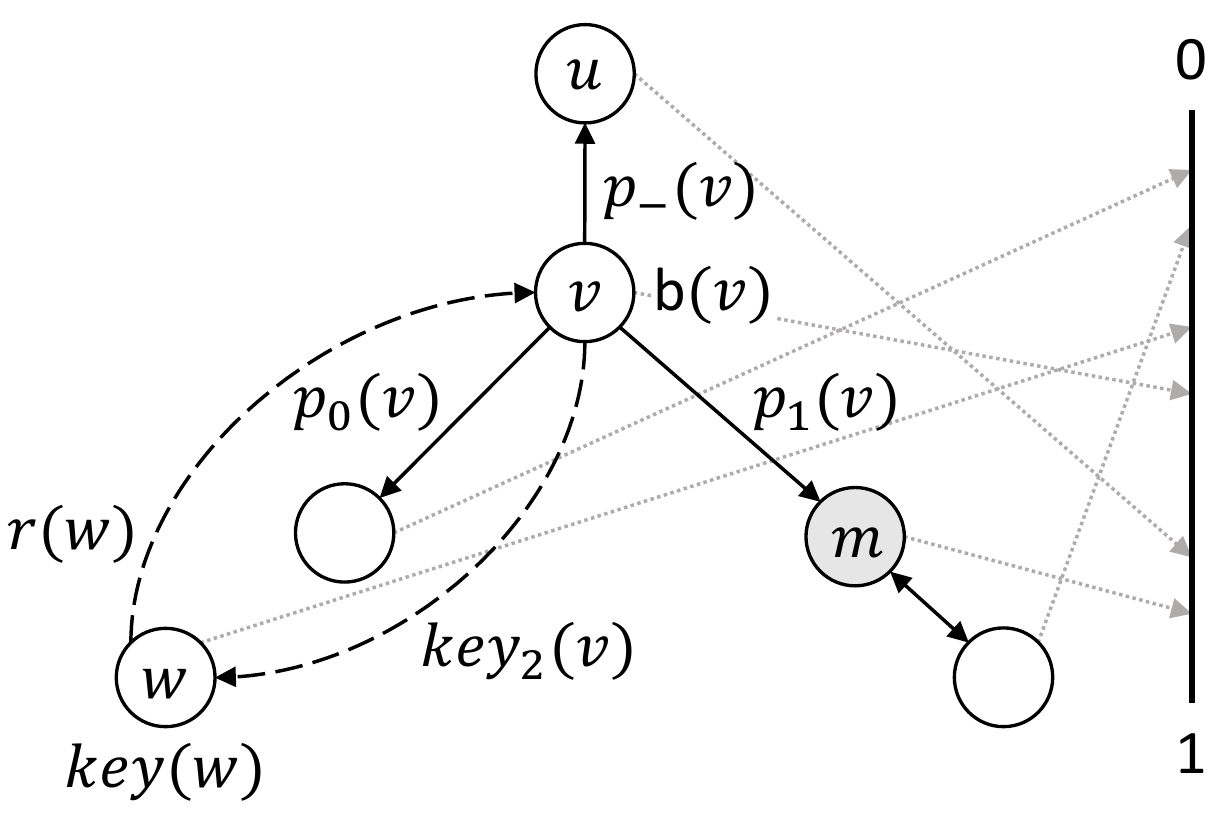}
		\captionof{figure}{Values stored at nodes of the \HPT\ from the perspective of $v$.
		Nodes are stored by hashing their label to $[0,1)$ in combination with a \DHT.
		White nodes denote Patricia nodes while Msd nodes are depicted in gray.}
		\label{figure:TRIE_Node_Values}
	\end{minipage}
\end{figure}

Every Patrica node $v$ has a label denoted by $b(v)$ and stores three edges.
The \emph{root} node stores the empty string $b(root)=\varepsilon$.
$p_{-}(v)$ is the parent edge of $v$ pointing to the parent node $u$ such that $b(u)\concat p_{-}(v)=b(v)$.
We denote by $\concat$ the concatenation of strings.
By $p_{x}(v)$ we denote the child edge of $v$ starting with the value $x$ for $x\in\{0,1\}$.
If $b(w)\in\KeySet$ for a Patricia node $w$, we set $key(w)=b(w)$.
Additionally, an inner Patricia node stores a $key_{2}(v)=k$, where $k$ is a key with $b(v)\sqsubseteq k$.
For efficient updates, the node $w$ storing $k$ has a field $r(w)=b(v)$.
These key$_2$ values allow returning a valid result for a prefix search when stopping at any Patricia node.
It is possible to assure that every inner Patricia node with two children has a key$_2$ pointing to a leaf node in its subtree.

To allow efficient prefix search, the Patricia Trie has been extended in \cite{HashedPatriciaTrie}.
Between every pair of directly connected Patricia nodes, Msd nodes (from Most Significant Digit) are added.
Their length is chosen in a way that those nodes are hit by a binary search first.
More specifically, Msd nodes are inserted between Patricia nodes such that their length is considered first by the binary search before the Patricia nodes around them are considered.
We only give a short definition of the calculation of an Msd label in \cref{definition:msd_calculation}.
In the special case that an Msd label equals the label of a surrounding Patricia node, no Msd node is needed at that position.
For details on how Msd nodes improve the prefix search operation, see \cite{HashedPatriciaTrie}.

\begin{definition}[Msd Label]
	\label{definition:msd_calculation}
	Let $a=(a_{m},\dots,a_0)$ and $b=(b_{m},\dots,b_0)$ be two binary strings of the same length.
	Possibly, one of them is filled up with leading zeros to have length $m+1$.
	We define $msd(a,b)$ to be the position $j$ where $a_j\neq b_j$ and $a_i=b_i$ for all $i>j$.
	That means, $msd(a,b)$ is the most significant bit (digit) at which $a$ and $b$ differ.
	\\
	Consider the binary labels $b(u)$ and $b(v)$ of two nodes $u,v$.
	Let $\ell_u=\vert b(u)\vert$ and $\ell_v=\vert b(v)\vert$ and without loss of generality let $\ell_u<\ell_v$.
	We define the Msd label $b(m)$ between $u$ and $v$ to be the prefix of $v$ of length $\sum_{i=msd(\ell_u,\ell_v)}^{\lfloor \log \ell_v\rfloor+1}(\ell_v)_{i}\cdot 2^{i}$.
\end{definition}

For example, consider $u,v$ with $b(u)=10$ and $b(v)=100101$, where $\ell_u=|b(u)|=(10)_{2}$ and $\ell_v=|b(v)|=(110)_{2}$.
Then $msd(\ell_u,\ell_v)=msd((010)_2,(110)_2)=2$, such that an Msd node $m$ between $u$ and $v$ has label $b(m)=1001\sqsubset b(v)$ of length $2^2=4$.

The \HPT\ supports operations $\textsc{PrefixSearch}(x)$ and $\textsc{Insert}(x)$ for a binary string $x$ in $\mathcal{O}(\log\vert x\vert)$ read accesses on the hash table.
Insertion takes additional $\mathcal{O}(1)$ write accesses and $\textsc{Delete}(x)$ is supported in constant hash table accesses.
Furthermore, the memory space usage is in $\Theta\left(\sum_{k\in\KeySet}\vert k\vert\right)$.

$ $\medskip\\
\textbf{Modification.}
We modify the \HPT\ to simplify the stabilization technique.
Consider \cref{figure:Modified_HPT}.
The original \HPT\ has a structure as shown on the left side.
The Msd node $m$ is in between the Patricia nodes $u$ and $w$ such that $u$ and $w$ point to $m$ and $m$ points to $u$ (parent) and $w$ (child).
We modify this structure by having $u$ and $w$ point to each other and not to $m$.
By this, deletions of Msd nodes do not concern the connectivity between Patricia nodes while the advantages of Msd nodes are still present.
The crucial property of Msd nodes is that they point to Patricia nodes.
Edges towards Msd nodes are not needed for the efficient operations introduced in \cite{HashedPatriciaTrie}.
For the rest of this paper, when we refer to the \HPT, we mean the \HPT\ with this small modification.
\\

\begin{wrapfigure}{r}{0.3\textwidth}
	\vspace{-0.7cm}
	\includegraphics[width=0.3\textwidth, trim = 0cm 13.5cm 27.5cm 0cm, clip=true]{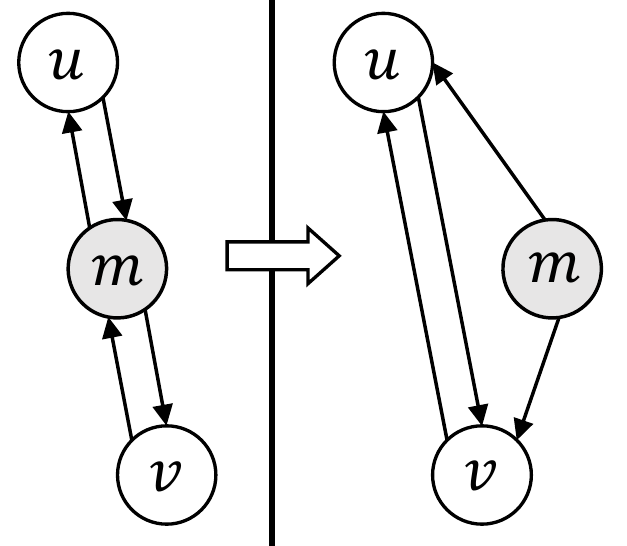}
	\caption{Modified \HPT}
	\label{figure:Modified_HPT}
	\vspace{-0.5cm}
\end{wrapfigure}
Next, we introduce some common terms that are used throughout the paper.
$\TrieSet$ is the set of all data nodes of the \HPT.
This includes $\PatSet$ as the set of nodes used in the original Patricia Trie and $\MsdSet$ which are the Msd nodes.
By definition $\TrieSet=\PatSet\cup\MsdSet$.
We denote by $\KeySet$ the set of keys stored by the \HPT\ and by $\MsgSet$ the set of all messages currently existing.
Let $u,\,v\in\TrieSet$ with $b(u) \sqsubset b(v)$.
In this case we say, $u$ is \emph{above} $v$ while $v$ is \emph{below} $u$.
Let $w\in\TrieSet$ such that $b(u) \sqsubset b(w) \sqsubset b(v)$.
Then $w$ is in \emph{between} $u$ and $v$.
If for two $u,\,v\in\TrieSet$ with $b(u) \sqsubset b(v)$ there is no $w\in\TrieSet$ with $b(u)\sqsubset b(w)\sqsubset b(v)$, then $u$ and $v$ are \emph{closest} to each other.
We say a child edge $e$ of $v\in\TrieSet$ is \emph{valid}, if there exists a node $w\in\TrieSet$ with $b(v) \concat e = b(w)$.
Similar, a parent edge $e$ of $v\in\TrieSet$ is valid, if there exists a node $w\in\TrieSet$ with $b(w) \concat e = b(v)$.
Consider two nodes $v,\,u\in\TrieSet$, where $u$ has an edge pointing to $v$ and vice versa.
We then speak of a \emph{bidirectional} edge.

\section{The \PN\ Protocol}\label{section:SHPT_protocol}

In the following, we present \PN, our self-stabilizing protocol for maintaining a \HPT.
The corrections of \PN\ can be divided into several parts.
We present our assumptions concerning the underlying \DHT\ first.
Afterwards, we give an intuition on the different types of repairs our protocol performs.
We often speak about actions executed by a \HPT\ node $v$.
This translates to actions that are executed by the corresponding \DHT\ node storing $v$.
For detailed Pseudocode, we refer to \cref{appendix:pseudo-code}.

\subsection{Properties of the DHT}

We assume that the underlying \DHT\ is in a legal state, i.e., it provides the actions \textsc{DHT-Search}$(x)$ and $\textsc{DHT-Insert}(x)$ which are carried out reliably on the stored data.
Deletion of data is only done locally by our protocol.
Stability of the \DHT\ is crucial as our protocol relies on finding/manipulating nodes of the \HPT\ solely based on their hash value given by their label.
There are a lot of different self-stabilizing \DHT s presented in the literature.
Some of them are mentioned in \cref{subsection:related_work}.

Our main demand on the \DHT\ is that at some point nodes are stored such that they can always be retrieved by their labels.
HPT nodes are essentially data-items.
Every \DHT\ node regularly checks if all its stored data is at the correct peer based on the hashing.
If data is stored incorrectly, it is sent towards the correct \DHT\ node.
When a data item $i$ is inserted at a \DHT\ node $n$, $n$ checks if $i$ is already present.
If yes, $i$ is only inserted if it does not collide with an already stored Patricia node that stores a key.
If a \HPT\ node $v$ has been inserted, a presentation method is triggered for $v$ and $v$ is directly presented to the nodes referred to by $p_{-}(v)$, $p_{0}(v)$ and $p_{1}(v)$.
The presentation mechanism is presented later.
This assumption assures that keys are preserved while insertion is not blocked and every \HPT\ node is presented at least once.

\subsection{Correcting Edge Information}

One general problem for self-stabilizing solutions is that every stored information can be corrupted.
Thus, our protocol regularly checks information stored in a \HPT\ node.
Consider a node $v\in\TrieSet$.
We refer to the information provided by the fields $p_{-}(v)$, $p_1(v)$ and $p_0(v)$ as well as $key_2(v)$ and $r(v)$ as \emph{edge information}.
Edge information can be checked rather simply as it allows reconstruction of a node's label $b(w)$.
The label can be used to query the \DHT\ for an (incomplete) copy of $w$.
$v$ can then compare the information stored at $w$ with its own and decide for corrections.
Some inconsistencies in the local structure can also be checked without querying the \DHT.
In general, when checking an edge $e$ at node $v$, we distinguish three cases:
\begin{enumerate}
	\item $e$ has a wrong form.
	      For example, if $p_{1}(v) = (0 \dots)$ or $p_{-}(v)$ is not a suffix of $b(v)$.
	      In this case, the edge is considered corrupted and is cleared.
	\item The node $w$ that $e$ points to does not exist.
	      Again, $e$ is not correct and is cleared.
	\item The node $w\in\TrieSet$ that $e$ points to does exist, but the edge provided by $w$ which should point to $v$ does not match $e$.
	      Several sub-cases arise here.
	      The protocol may have to simply present $v$ to $w$, or a new node may need to be inserted.
\end{enumerate}

\begin{figure}[htb]
	\center
	\includegraphics[width=0.95\textwidth, trim = 0cm 16cm 15.5cm 0cm, clip=true]{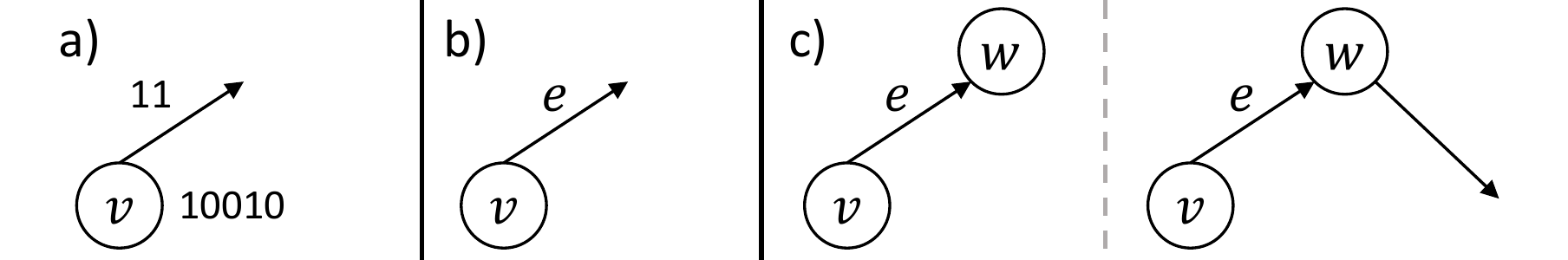}
	\caption{Examples for the cases of wrong edge information.}
	\label{figure:Cases_edge_information}
\end{figure}

Additionally, every node avoids edges pointing to Msd nodes.
Such edges are treated as if they pointed to a non-existing node.
A node $v$ can check the values of $p_{-}(v)$, $key_2(v)$ and $r(v)$ by calculating if the prefix relation between itself and the respective nodes fulfills the definition of the hashed Patricia Trie.
To prevent the spreading of incorrect information, new edges are only stored if they comply with the definition of the hashed Patricia Trie from the local perspective of $v$.
We will go into detail on the creation of new edges and the insertion of nodes later.

\subsection{Maintaining Connections}

Our goal to stabilize the Patricia nodes of a \HPT\ can also be formulated using \emph{Branch Sets} as described in \cref{definition:branch_Set}.
A Branch Set consists of all Patricia nodes on a branch from the root to a leaf node (see \cref{figure:Presentation_List_Branch}).
When the \HPT\ is in a legal state, there are as many Branch Sets as there are leaf nodes.

\begin{restatable}[Branch Set]{definition}{definitionBranchSet}
	\label{definition:branch_Set}
	Consider a set of Patricia nodes with maximum cardinality $S$ such that $u,w\in S$ implies $b(u)\sqsubset b(w)$ or $b(w)\sqsubset b(u)$ and the Patricia node $v\in S$ with maximum label length stores a key $k$.
	We call this set the \emph{Branch Set} of $k$.
\end{restatable}

\begin{wrapfigure}[21]{r}{0.4\textwidth}
	\vspace{-0.3cm}
	\includegraphics[width=0.4\textwidth, trim = 0cm 10.5cm 25.5cm 0cm, clip=true]{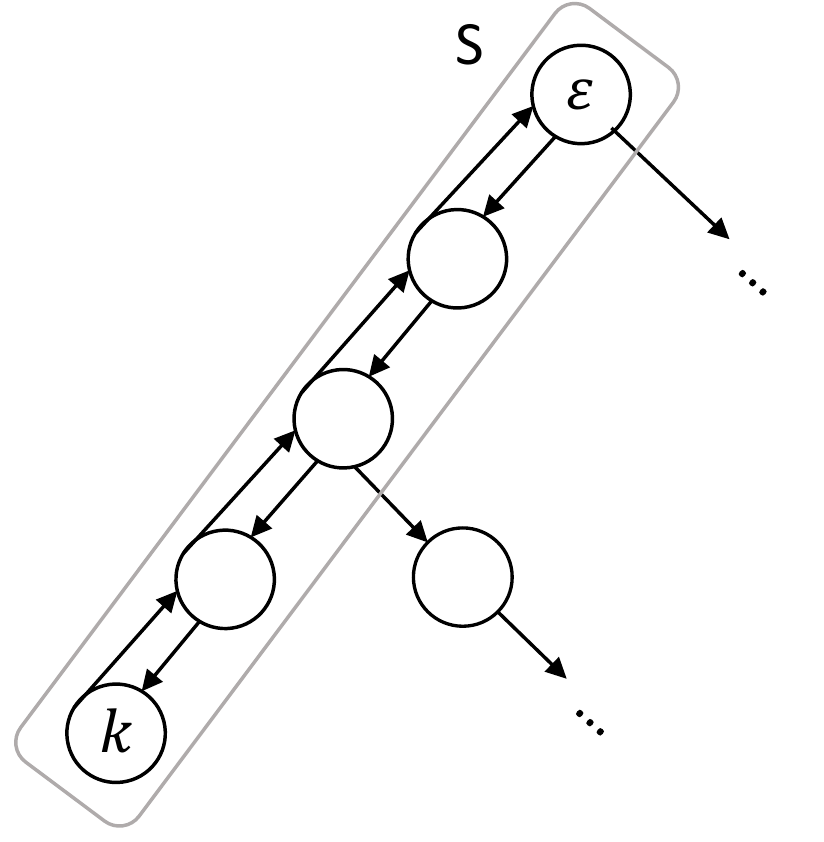}
	\caption{Branch Set $S$ from the root ($\varepsilon$) to a leaf node ($k$) is the set of nodes in a branch of the hashed Patricia Trie in a legal state.}
	\label{figure:Presentation_List_Branch}
\end{wrapfigure}
We apply a technique called Linearization \cite{Onus:2007:LLS:2791188.2791198} to all Patricia nodes to create a list sorted by label length for all Branch Sets in finite time.
It is important to exclude Msd nodes from the Linearization.
Msd nodes are not presented nor do they delegate presentation messages.
Due to deletion of a Patricia node, an Msd node might still be presented accidentally.
However, we limit this problem by carefully handling deletions and insertions as described later.
For the Linearization to work, we need to make sure that all nodes in a Branch Set are brought into and kept in a weakly connected state.

A Patricia node $v$ with an empty parent edge tries to recreate connectivity by doing a modified \textsc{PrefixSearch}($b(v)$) similar to the one presented in \cite{HashedPatriciaTrie}.
The procedure we call \textsc{BinaryPrefixSearch}($b(v)$) does not search for $b(v)$ itself and only consists of the binary search phase of the PrefixSearch($x$) of \cite{HashedPatriciaTrie}, returning a copy of a Patricia node $w$ with $b(w)\sqsubset b(v)$.
If no such node exists, we conclude that the root node is non-existent and trigger a construction of it.

Further, we let every Patricia node present its own label to its parent and its two children using a \emph{presentation message}.
A message presenting $v$ is delegated to the Patricia node $w$ closest to $v$.
Delegation happens only by using edges and intermediate nodes sharing a Branch Set with $v$.
All nodes maintain connections to labels which are closest to them while delegating presentations of other labels.
This behavior resembles the Linearization approach presented in \cite{Onus:2007:LLS:2791188.2791198}, allowing our protocol to form a sorted list for all branches of the \HPT.

There is still an important issue we need to resolve.
Consider a Branch Set $S$ of nodes.
We can end up in situations where nodes exist that do not contribute to the hashed Patricia Trie.
Such nodes can be Patricia nodes not storing a key.
To reduce memory demands, we are interested in removing unneeded nodes.
In principle, deletion without harming connectivity can be done since the root node is always known implicitly.
However, deletion increases distances.
In addition, our protocol must provide the ability to create and integrate new Patricia nodes.
When inserting and deleting nodes, we need to make sure that no loops are possible in which the system may take forever to stabilize.
We will explain how to avoid such loops in the following.

\subsection{Removal/Creation of Nodes}

Due to the implicitly known root node, deletion is possible and should be considered to reduce memory demands.
We distinguish between Msd nodes and Patricia nodes.
Our modification allows us to handle Msd nodes in a simple and efficient way.
We try to avoid any edges pointing to Msd nodes such that eventually, deletion and creation of Msd nodes does not influence the Patricia nodes and their structure.
Only if there are two Patricia nodes $u,\,w$ connected via a bidirectional edge, an Msd node between them might be inserted.
Fortunately, Msd labels can be calculated locally and a corresponding Msd node can easily be accessed by querying the \DHT.
Any Msd node which is not between such two Patricia nodes, or has an incorrect label, is deleted.

A Patricia node $v$ (except for the root) is \emph{unnecessary} if $key(v)=nil$ and there are no two Patricia nodes $u,w$, both storing a key, such that $b(v) = \ell cp(b(u), b(w))$, i.e., $u$ should be in a different subtree than $w$ below $v$.
From a global point of view, we can easily decide if $v$ is unnecessary solely based on information about the situation below $v$.
From a local perspective, $v$ cannot decide but only assume to be unnecessary if it lacks child edges.
We make the local protocol aggressive by deleting any node that lacks child edges and assumes to be unnecessary.
This also introduces deletion of necessary Patricia nodes.
Therefore, we always trigger a creation of new \HPT\ nodes by Patricia nodes below the new ones.
This avoids loops of creation and deletion of nodes, because newly created nodes inherently have valid children and, thus, do not assume to be unnecessary.
Patricia nodes storing a key essentially form a stable starting point, because they are never deleted.
The need to insert a Patricia node is detected by comparing a node's parent edge with the corresponding edge provided by the parent.

\subsection{Distribution of References to Keys}

In addition, \PN\ tries to achieve the following.
Every inner Patricia node $v$ with two children should store a $key_{2}(v)=b(w)$ which points to a leaf node $w$ storing a key such that $b(v)\sqsubset b(w)$.
The respective leaf node $w$ stores an $r(w)$ value pointing to $v$.
This property is helpful for efficient prefix search.
No matter at which Patricia node the prefix search stops, there is a key referenced having the node's label as a prefix.
This key is a valid result for the search query.
We call all inner Patricia nodes with two children and the root node \emph{key$_2$ nodes}.
Due to the resemblance of the hashed Patricia Trie with a binary tree, \cref{fact:enough_key2_spaces_for_leaves} holds.

\begin{restatable}{fact}{factEnoughKeyTwoSpacesForLeaves}
	\label{fact:enough_key2_spaces_for_leaves}
	Let $L$ be the number of leaf nodes.
	Let $I$ be the number of key$_2$ nodes.
	When the \HPT\ is in a legal state, it holds $I \leq L \leq I +1$.
	$L=I$, if the root has one child and $L=I+1$ if it has two.
\end{restatable}

To assure that every leaf node is referenced by a key$_2$ node, we allow the root to store up to two key$_2$ values.
This reduces the number of hash table accesses created by our protocol, when the \HPT\ is in a legal state.\\

If we naively assign leaf nodes to key$_2$ nodes, this may lead to situations in which a key$_2$ node cannot get a key$_2$ value.
For an example, consider \cref{figure:Blocking_Example}.
The critical observation is that key$_2$ nodes with a shorter label, in general, have more possible leaf nodes they can point to than key$_2$ nodes with a longer label.
Therefore, our protocol aims at prioritizing key$_2$ nodes which are closer to leaf nodes.

\begin{figure}[htb]
	\begin{minipage}[c]{0.6\textwidth}
		\centering
		\includegraphics[width=\textwidth, trim = 0cm 14cm 21.5cm 0cm, clip=true]{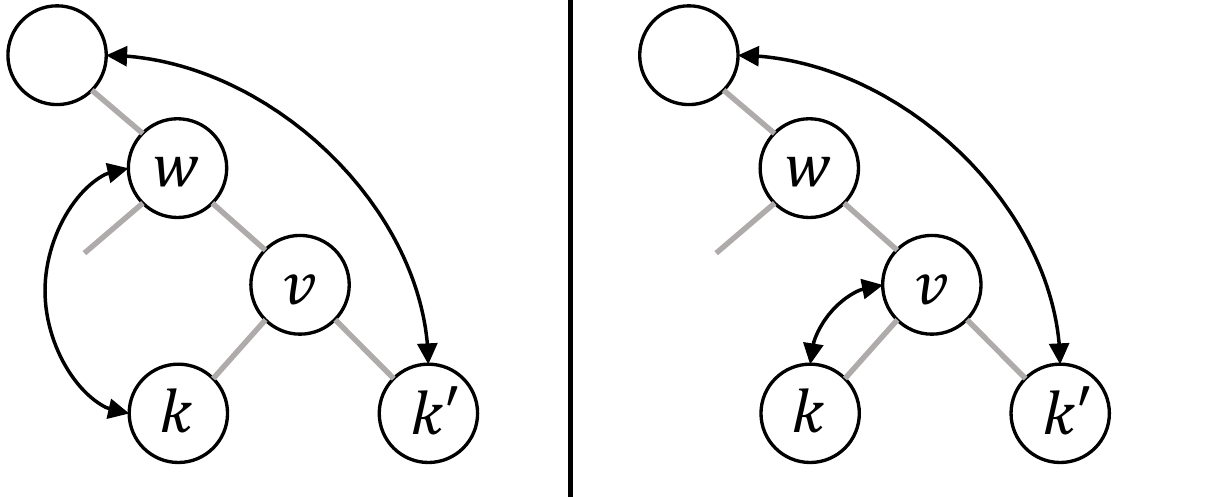}
	\end{minipage}
	\begin{minipage}[c]{0.4\textwidth}
		\caption{Example where $v$ cannot get a key$_2$ (left).
			The leaf nodes $k$ and $k'$ storing a key are already associated to Patricia nodes above $v$.
			The blocking of $v$ is resolved as $v$ takes over the key$_2$ of $w$ (right).}
		\label{figure:Blocking_Example}
	\end{minipage}
\end{figure}

We divide the protocol into three parts.
First, all nodes continuously check if they should store a key$_2$ or $r$ value and whether such a value points to a leaf node, respectively key$_2$ node.
Second, if a leaf node $v$ does not store a value in $r(v)$, it presents its label upwards in the \HPT\ by sending a message crossing only parent edges.
The first key$_2$ node $w$ without a key$_2$ receiving the message sets $key_2(w)=b(v)$.
Third, a key$_2$ node $v$ repairs in the following way.
If $key_2(v)$ points to leaf node $w$ with $b(v)\sqsubset b(w)$, there are two cases.

\begin{enumerate}
	\item $b(v)\sqsubset r(w)$: Then $key_2(v)$ is set to $nil$ since there may already be some key$_2$ node with longer label pointing at $w$.
	\item Else, $v$ has either longer label than $r(w)$ or $r(w)=nil$.
	      The protocol sets $r(w)=b(v)$.
\end{enumerate}

If $key_2(v)=nil$, a message is sent upwards in the \HPT\ and the first key$_2$ node $w$ with $b(v)\sqsubset key_2(w)$ responds to $v$.
Then, $key_2(v)$ is set to $key_2(w)$.
Eventually, $v$ takes over the key$_2$ value of $w$, because $w$ executes case a).

Intuitively, key$_2$ nodes without a key$_2$ pull values from nodes with shorter label.
Simultaneously, leaf nodes without an $r$ value present their label towards the root.

\section{Protocol Analysis}\label{section:Protocol_analysis}

In this section, we show that \PN\ is self-stabilizing and transforms the \HPT\ in finite time to a legal state.
Furthermore, we present results concerning memory usage and the number of hash table accesses and messages when the \HPT\ is in a legal state.

\subsection{Correctness}\label{section:protocol_analysis:correctness}

We begin by showing the correctness of our self-stabilizing protocol.
We use a commonly known technique introduced by Dijkstra in \cite{Dijkstra:1974:SSS:361179.361202}.
Our goal is to show \cref{theorem:HPTCorrect}.
For that we consider a sequence of intermediate states that are reached consecutively until the \HPT\ is in a legal state.
For every state we show \emph{convergence} towards the state and \emph{closure} within it, i.e., the properties of the state are kept by our protocol.

\begin{restatable}{theorem}{CorrectnessTheorem}
	\label{theorem:HPTCorrect}
	The algorithm creates in finite time a hashed Patricia Trie in a legal state out of any initial state in which the \DHT\ is in a legal state and there is a set of unique keys stored at \DHT\ nodes.
\end{restatable}

In the following, we briefly sketch the main proof by presenting a sequence of main lemmas that roughly reflect the states the system reaches.
Each main lemma thereby consists of multiple properties that are proven by a set of lemmas on its own.
The full proof consisting of all lemmas, their respective proofs, and the complete definition of a legal state of the \HPT\ can be found in \cref{appendix:correctness_proof}.

To prove the correctness captured in \cref{theorem:HPTCorrect}, we first need to formally define a legal state of the \HPT.
At this point, we only give an intuitive definition.
For the complete definition, see \cref{appendix:correctness_proof} \cref{definition:Correctness_of_HPT}.
Intuitively, the \HPT\ is in a legal state if we have as few \HPT\ nodes as possible in the system, all keys are stored correctly, the structure is consistent to the (modified) definition presented in \cref{section:HashedPatriciaTrie}, and the references to keys in key$_{2}$ nodes are existing and stored at correct nodes.

Initially, we only assume that a set of unique keys is stored at \DHT\ nodes.
The first lemma states that general repair mechanisms assure correctly stored keys and Patricia nodes.

\begin{restatable}{lem}{lemmaMainMiscRepair}
	\label{lemma:main:misc_repair}
	In finite time it holds:
	Every key $k$ is stored in a node $v\in\PatSet$ with $b(v)=k$.
	Furthermore, every node is stored at the \DHT\ node responsible for it.
	Consider any $v\in \TrieSet$ that is deleted.
	As long as $v$ is not reconstructed, in finite time it holds:
	\begin{enumerate}
		\item There is no presentation message for $b(v)$.
		\item There is no edge pointing towards $b(v)$ in the system.
	\end{enumerate}

\end{restatable}

From now on, the proof consists of three phases.
In a first phase, all Patricia nodes which are not needed for the final structure are removed.
The second phase considers the reconstruction of the binary tree structure of the \HPT\ and corrects the sets of Patricia nodes and Msd nodes.
In the third and last phase, information stored in key$_2$ and $r$ fields is made consistent.

\subsubsection{Phase \rom{1} -- Deletion of Patricia nodes}

In this phase, the protocol makes sure that all Patricia nodes which are not needed in the final structure are removed.
Initially, information stored at \HPT\ nodes that directly contradicts the definition of the \HPT\ is cleared.
This can be information such as a parent edge at $v\in\TrieSet$ that is no suffix of $b(v)$.
After that, Patricia nodes and Msd nodes in unnecessary subtrees, i.e., subtrees not containing a key, and unnecessary inner Patricia nodes are gradually removed.
Every leaf node in an unnecessary subtree detects in finite time that it has no valid children and is deleted.

\begin{restatable}{lem}{lemmaMainEmptySubtreesNotNeededNodesRemoved}
	\label{lemma:main:Empty_Subtrees_Not_Needed_Nodes_Removed}
	In finite time, every unnecessary Patricia node is removed.
	A Patricia node $v$ is \emph{unnecessary} if there are no two $keys$ $k_1$ and $k_2$ with $b(v)=\ell cp(k_1, k_2)$.
\end{restatable}

\begin{figure}[htb]
	\begin{minipage}[c]{0.68\textwidth}
		\centering
		\includegraphics[width=\textwidth, trim = 0cm 11.5cm 17.4cm 0cm, clip=true]{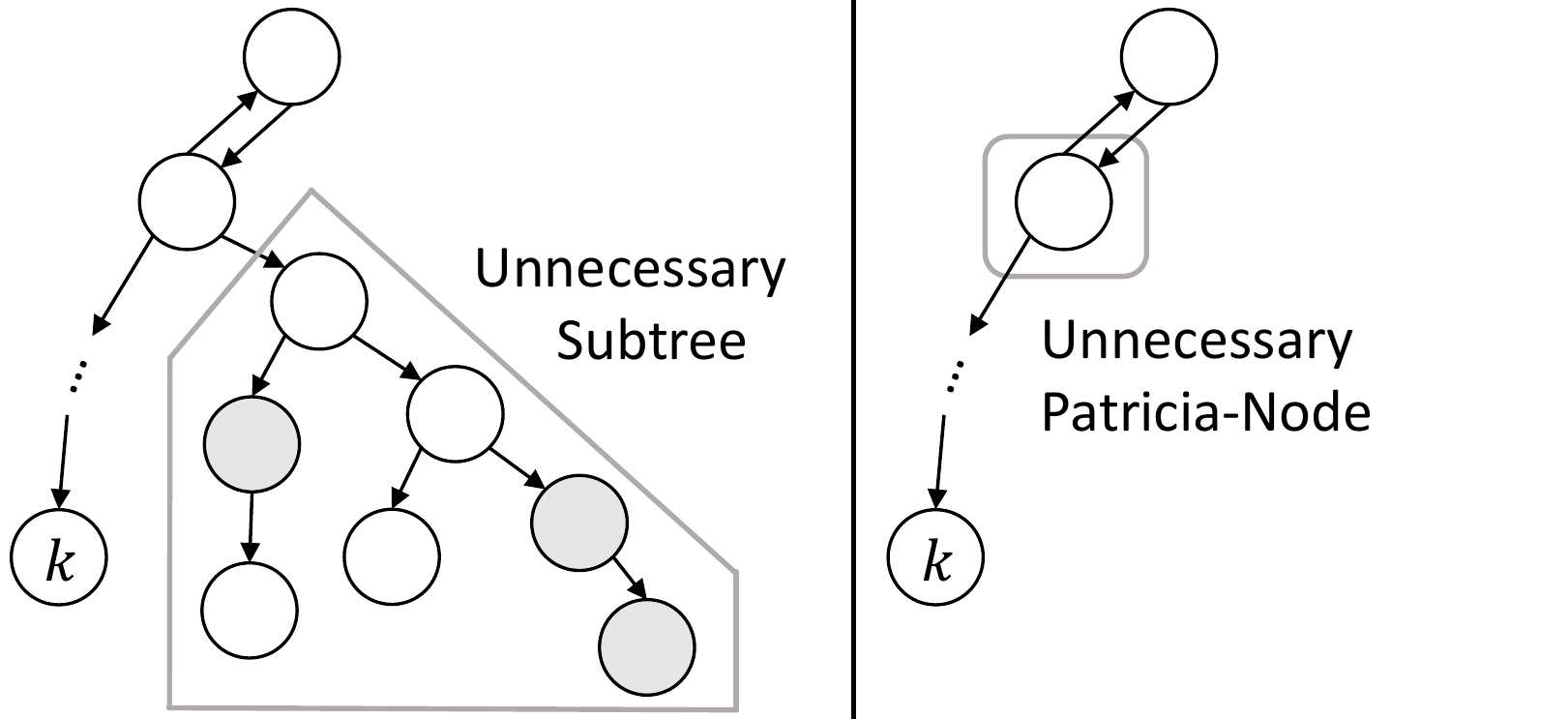}
	\end{minipage}
	\begin{minipage}[c]{0.28\textwidth}
		\caption{Node $k$ stores a key. Msd nodes are sketched in grey. First, unnecessary subtrees are deleted (left), then remaining unnecessary Patricia nodes are removed (right).}
		\label{figure:Deletion_Phase}
	\end{minipage}
\end{figure}

Patricia nodes which are necessary may still be deleted because of their local perspective.
However, this deletion is limited and stops after finitely many deletions.
This holds, because Patricia nodes are only deleted due to incorrect child edges.
If a new Patricia node with a long label is inserted, its child edges are initially valid and stay valid.
There cannot be infinitely many deletions triggered, because the structure stabilizes bottom-up.

\begin{restatable}{lem}{lemmaMainDeletionStops}
	\label{lemma:main:Deletion_Stops}
	In finite time, every Patricia node has valid child edges pointing to Patricia nodes and no further Patricia node is deleted.
\end{restatable}

\subsubsection{Phase \rom{2} -- Reconstruction}

In the second phase, \PN\ reconstructs the \HPT\ by rebuilding missing Patricia nodes and repairing connections.
Since every node tries to create a parent edge pointing to a Patricia node with shorter label, eventually all missing Patricia nodes are detected and can be inserted.
The process works in a bottom-up fashion, i.e., Patricia nodes with longer labels reconstruct missing nodes with shorter ones.
The Patricia nodes storing a key as well as the root node act as fixed points in this case, because they are never deleted once constructed.

\begin{restatable}{lem}{lemmaMainReconstructionRootCorrectPointing}
	\label{lemma:main:Reconstruction_Root_Correct_Pointing}
	In finite time, the root node exists and no Patricia node points to an Msd node.
	Furthermore, missing Patricia nodes are reconstructed.
	Also, every Patricia node has valid edges pointing only to existing Patricia nodes, i.e., there is a path from every Patricia node to the root and there is a path from the root to every Patricia node.
\end{restatable}

It is crucial that no Patricia node points to an Msd node, because edges to Msd nodes are effectively treated as corrupt ones.
This property assures that Msd nodes are eventually excluded from the Linearization procedure.
Linearization then allows us to show that every Branch Set (see \cref{definition:branch_Set}) of Patricia nodes eventually forms a stable sorted list.
Incorrect Msd nodes are removed without affecting the rest of the \HPT\ and missing Msd nodes are inserted.
Further, correct Msd nodes are not deleted, because the two Patricia nodes closest to a correct Msd node are not deleted and do not change their edges any more.
All these properties are reflected in \cref{lemma:main:Linearization_Works_Msd_Nodes_Removed_And_Inserted}.
For completeness, we refer to the definition of incorrect and missing Msd nodes in the full proof in \cref{appendix:correctness_proof}.

\begin{restatable}{lem}{lemmaMainLinearizationWorksMsdNodesRemovedAndInserted}
	\label{lemma:main:Linearization_Works_Msd_Nodes_Removed_And_Inserted}
	In finite time for every Branch Set $S$ it holds:
	Between every pair of closest Patricia nodes $u,\,w\in S$ there is a bidirectional edge.
	Furthermore, every incorrect Msd node is removed and all missing Msd nodes are inserted.
\end{restatable}

\subsubsection{Phase \rom{3} -- Consistency}

In the final phase the information stored in key$_2$ and $r$ fields is corrected to be consistent.
Due to \cref{fact:enough_key2_spaces_for_leaves}, we know that this can be achieved.
The root is allowed to store up to two key$_2$ values.
Therefore, there is always a way to store all keys of leaf nodes in key$_2$ nodes.
First, we show that nodes which should not store a key$_2$ value remove any such stored value.
Further, references in key$_2$ and $r$ fields are deleted when they contradict the relationship $r(key_2(v))=b(v)$, where $v$ is a key$_2$ nodes and $key_2(v)$ references a leaf node.

\begin{restatable}{lem}{lemmaMainKeyTwoRValueReferencesOnlyEligible}
	\label{lemma:main:Key_Two_R_Value_References_Only_Eligible}
	In finite time, only key$_2$ nodes store a key$_2$ and only leaf nodes store an $r$ value.
	Every key$_2$ value stored at a Patricia node $v$ points to a leaf $w$ with $b(v)\sqsubset b(w)$ and every $r$ value stored at a Patricia node $w$ points to a key$_2$ node $v$ with $b(v)\sqsubset b(w)$.
\end{restatable}

From now on, key$_2$ nodes not storing a key$_2$ try to acquire the key$_2$ of a key$_2$ node above them.
Leaf nodes lacking a reference in $r$ present themselves to key$_2$ nodes above them.
Therefore, the length of the longest label of a key$_2$ node not storing a staying key$_2$ reduces over time.
As this length is finite, the process terminates.
Thereafter, the $r$ values of leaf nodes are corrected, because the key$_2$ values do not change any more.

\begin{restatable}{lem}{lemmaMainKeyTwoRValueReferencesMadeConsistent}
	\label{lemma:main:Key_Two_R_Value_References_Made_Consistent}
	In finite time, all key$_2$ nodes store a stable key$_2$ and all leaf nodes store a stable $r$ value.
	For every key$_2$ node $v$, the node $w$ with $b(w)=key_2(v)$ is a leaf node with $r(w)=b(v)$.
\end{restatable}

Finally, our protocol is correct as all unnecessary nodes are removed, missing nodes are inserted, Patricia nodes are connected by bidirectional edges, and the information stored in key$_2$ and $r$ fields is consistent such that the \HPT\ is is in a legal state in finite time.

\subsection{Overhead}

Assume, the \HPT\ is in a legal state.
We give results for the complexity in terms of hash table accesses and messages and the memory overhead of our solution.
We refer to \cref{appendix:overhead_proof} for the proofs of the following theorems.
When a \DHT\ node executes \PN\ by calling its \textsc{Timeout} Method, exactly one \HPT\ node is checked.
Thereby, at most a constant number of other \HPT\ nodes may be partially acquired or notified and \cref{theorem:Constant_Amount_Of_Messages_Per_Timeout} holds.

\begin{restatable}{theorem}{theoremConstantAmountOfMessagesPerTimeout}
	\label{theorem:Constant_Amount_Of_Messages_Per_Timeout}
	When the \HPT\ is in a legal state, \PN\ creates a constant number of hash table (read) accesses and messages per call of \textsc{Timeout} at each \DHT\ node.
\end{restatable}

Unnecessary Patricia nodes and incorrect Msd nodes are removed by \PN.
Therefore, the \HPT\ nodes are the same as presented in the construction in \cref{section:HashedPatriciaTrie} and \cref{theorem:storage_overhead} holds.

\begin{restatable}{theorem}{theoremStorageOverhead}
	\label{theorem:storage_overhead}
	Let $\bitsSym$ be the number of bits needed to store all keys.
	The total memory used by a \HPT\ in a legal state is in \MemoryBound\ bits.
\end{restatable}

\bibliography{bib}

\newpage

\appendix

\section{Pseudocode}\label{appendix:pseudo-code}

In the following, we present the Pseudocode for \PN.
For simplicity, we assume the following functions to exist.

\begin{table}[H]
	\center
	\setlength{\tabcolsep}{5pt}
	\renewcommand{\arraystretch}{1.5}
	\begin{tabularx}{\textwidth}{|l|X|}
		\hline
		\textbf{Function Name}                                 & \textbf{Effect when called}                                                                                                                                                         \\\hline
		\textsc{DHT-Search}($x\in\{0,1\}^{*}$)                 & Returns a copy of the \HPT\ node $v$ with $b(v)=x$.                                                                                                                                 \\\hline
		\textsc{DHT-Insert}($x$)                               & Stores the data-item $x$ in the \DHT.                                                                                                                                               \\\hline
		$\textsc{Parent}(v\in\TrieSet)$                        & Provides the bit string $x$ such that $x\concat p_{-}(v)=b(v)$.                                                                                                                     \\\hline
		$\textsc{Children}(v\in\TrieSet)$                      & Returns the number of children of $v$.                                                                                                                                              \\\hline
		\textsc{Edge}($w\in\TrieSet,\,v\in\TrieSet$)           & Provides the bit string of the edge $e$ of $w$ that should point to $v$ if $v$ was $w's$ child node.                                                                                \\\hline
		\textsc{MsdMissing}($w\in\PatSet,\,v\in\PatSet$)       & Returns $true$ if there is a missing Msd node between $w$ and $v$.                                                                                                                  \\\hline
		\textsc{BinaryPrefixSearch}($x\in\{0,1\}^{*}$)         & Executes the binary search phase of the PrefixSearch($x$) algorithm provided by \cite{HashedPatriciaTrie} and returns a copy of the found Patricia node $w$ with $b(w)\sqsubset x$. \\\hline
		\textsc{MsdChildEdge}($v\in\MsdSet$)                   & Provides the only child edge of an Msd node.
		Returns $nil$ if $v$ has more than one child.
		An ordinary Msd node only has one child edge.                                                                                                                                                                                                \\\hline
		\textsc{MsdLabel}($x\in\{0,1\}^{*},\,y\in\{0,1\}^{*}$) & Returns a valid Msd label between the labels $x$ and $y$.                                                                                                                           \\\hline
		\textsc{Bidirectional}($v\in\TrieSet,\,w\in\TrieSet$)  & Returns $true$ if there is a bidirectional edge between $v$ and $u$.                                                                                                                \\\hline
	\end{tabularx}
	\noindent\vspace{0.4cm}
	\caption{Assumed existing functions}
\end{table}

\newpage
\begin{algorithm}[H]
	\caption{The \DHT\ protocol $\rightarrow$ executed at \DHT\ node $n$}
	\label{algorithm:DHT_node_protocol}
	\begin{algorithmic}[1]
		\Procedure{Timeout}{}
		\State $K\gets \text{set of keys not in Patricia nodes stored at }n$
		\If{$K\neq\emptyset$}\Comment{Integrate keys}
		\State Insert Patricia node $v$ for every key $k\in K$ with $b(v)=k$\label{algoithm:line:save_Key_In_Patricia_Node}
		\EndIf
		\State $v\gets \text{next \HPT\ node stored at }n$
		\If{$n \text{ not responsible for }b(v)$}\Comment{Check position in \DHT}
		\State delegate $v$ towards the correct \DHT\ node\label{algorithm:line:delegation_wrong_saved_info}
		\EndIf
		\State \textsc{CheckNodeInfo}($v$)\Comment{Consistency of stored information}
		\State \textsc{CheckParentEdgeInfo}($v$)\Comment{Check parent edge}
		\State \textsc{CheckChildEdgeInfo}($v$)\Comment{Check child edges}
		\State \textsc{CheckValidity}($v$)\Comment{Check if node is needed}
		\State \textsc{CheckKey$2$Info}($v$)\Comment{Check key$_{2}$ information}
		\State \textsc{LinearizeTimeout}($v$)\Comment{Do Linearization}
		\EndProcedure
	\end{algorithmic}
\end{algorithm}
\noindent
\cref{algorithm:DHT_node_protocol} calls different checks for one \HPT\ node $v$ on a regular basis.\\\\

\begin{table}[H]
	\setlength{\tabcolsep}{5pt}
	\renewcommand{\arraystretch}{1.5}
	\begin{tabular}{|l|l|}
		\hline
		\textbf{Line(s)} & \textbf{Affected Lemmas}                   \\\hline
		$2-4$            & \cref{lemma:keys_saved_at_right_node}      \\\hline
		$6-7$            & \cref{lemma:any_node_saved_correct_in_dht} \\\hline
	\end{tabular}
	\noindent\vspace{0.4cm}
	\caption{Lemmas affected by \cref{algorithm:DHT_node_protocol}}
\end{table}

\newpage
\begin{algorithm}[H]
	\caption{Checking information stored directly at a \HPT\ node}
	\label{algorithm:Trie-Node_local_check}
	\begin{algorithmic}[1]
		\Procedure{CheckNodeInfo}{$v\in\TrieSet$}
		\If{$p_{-}(v) \text{ is not suffix of }b(v)$}\label{algorithm:line:parent_edge_no_suffix}\Comment{Parent edge violates definition}
		\State $p_{-}(v) \gets nil$
		\EndIf
		\If{$p_{0}(v) \text{ does not comply } 0(0\vert 1)^{*}$}\label{algorithm:line:p0_wrong_form}\Comment{Child edge violates definition}
		\State $p_{0}(v)\gets nil$
		\EndIf
		\If{$p_{1}(v) \text{ does not comply } 1(0\vert 1)^{*}$}\Comment{Child edge violates definition}
		\label{algorithm:line:p1_wrong_form}
		\State $p_{1}(v)\gets nil$
		\EndIf
		\If{$v\notin\MsdSet$}
		\If{$key(v) \neq nil \textbf{ and } b(v) \neq key(v)$}\Comment{Wrong label}\label{algorithm:line:wrong_saved_key}
		\State $w\gets\text{new Patricia node}$
		\State $key(w)\gets key(v)$
		\State \textbf{delete} $v$
		\State \textsc{DHT-Insert}($w$)
		\Else
		\If{$key_{2}\neq nil$}
		\If{$b(v) \not\sqsubseteq key_{2}(v)$$\textbf{ or } (\textsc{Children}(v)<2 \textbf{ and } v \neq root)$}\label{algorithm:line:key2_wrong_form}
		\State $key_{2}(v)\gets nil$ \Comment{$v$ not a key$_2$ node}
		\EndIf
		\EndIf
		\If{$r(v) \neq nil \textbf{ and } r(v) \not\sqsubset b(v)$}\label{algorithm:line:r_wrong_form}
		\State $r(v)\gets nil$ \Comment{Wrong reference}
		\EndIf
		\If{$r(v) \neq nil \textbf{ and }( p_{0}(v) \neq nil \textbf{ or } p_{1}(v) \neq nil )$}\label{algorithm:line:r_value_stored_at_inner_node}
			\State $r(v)\gets nil$ \Comment{$v$ not a leaf node}
		\EndIf
		\EndIf
		\EndIf
		\EndProcedure
	\end{algorithmic}
\end{algorithm}
\noindent
\cref{algorithm:Trie-Node_local_check} locally checks if the information stored at \HPT\ node $v$ is consistent.
The algorithm removes information at $v$ which cannot be correct.\\\\

\begin{table}[H]
	\setlength{\tabcolsep}{5pt}
	\renewcommand{\arraystretch}{1.5}
	\begin{tabular}{|l|l|}
		\hline
		\textbf{Line(s)} & \textbf{Affected Lemmas}                                                                                                                               \\\hline
		$2-3$            & \cref{lemma:initially_wrong_reference_information_vanishes}                                                                                            \\\hline
		$4-5$            & \cref{lemma:initially_wrong_reference_information_vanishes}                                                                                            \\\hline
		$6-7$            & \cref{lemma:initially_wrong_reference_information_vanishes}                                                                                            \\\hline
		$8-13$           & \cref{lemma:keys_saved_at_right_node}, \cref{lemma:no_more_deletions_happen}                                                                           \\\hline
		$15-17$          & \cref{lemma:initially_wrong_reference_information_vanishes}, \cref{lemma:not_eligible_nodes_drop_key2}, \cref{lemma:strange_references_key2_disappear} \\\hline
		$18-19$          & \cref{lemma:initially_wrong_reference_information_vanishes}, \cref{lemma:not_eligible_nodes_drop_key2}, \cref{lemma:strange_references_key2_disappear} \\\hline
		$20-21$          & \cref{lemma:not_eligible_nodes_drop_key2}                                                                                                              \\\hline
	\end{tabular}
	\noindent\vspace{0.4cm}
	\caption{Lemmas affected by \cref{algorithm:Trie-Node_local_check}}
\end{table}

\newpage
\begin{algorithm}[H]
	\caption{Checking parent edge of $v\in\TrieSet$}
	\label{algorithm:Trie-Node_parent_edge_check}
	\begin{algorithmic}[1]
		\Procedure{CheckParentEdgeInfo}{$v\in\TrieSet$}
		\If{$p_{-}(v) = nil \textbf{ and } b(v) \neq \varepsilon \textbf{ and } v\notin\MsdSet$}
		\State $w \gets \textsc{BinaryPrefixSearch}(b(v))$
		\If{$w = nil$} \Comment{Root does not exist}\label{algorithm:line:root_creation}
		\State $root \gets \text{new Patricia node}$
		\State $b(root) \gets \varepsilon$
		\State \textsc{DHT-Insert}($root$)
		\Else
		\State $p_{-}(v) \gets x \text{ [s.t.\ } b(w) \concat x = b(v) ]$\Comment{New parent}\label{algorithm:line:new_parent_due_to_binary_prefix_search}
		\EndIf
		\ElsIf{$b(v) \neq \varepsilon \textbf{ and } v\notin\MsdSet$}
		\State $par \gets \textsc{DHT-Search}(\textsc{Parent}(v))$\Comment{Acquire parent}
		\If{$par = nil$\textbf{ or } $par \in \MsdSet$}\label{algorithm:line:non_existing_parent_node_edge_deleted}
		\State $p_{-}(v) \gets nil$
		\Else\Comment{Compare edges}
		\State $e_{par} \gets \textsc{Edge}(par,\,v)$
		\If{$e_{par} \neq p_{-}(v) \textbf{ and } e_{par} \not\sqsubseteq p_{-}(v) \textbf{ and } p_{-}(v) \not\sqsubseteq e_{par}$}\label{algorithm:line:parent_should_be_patricia_node_with_two_children}
		\State $n \gets \text{new Patricia node}$
		\State $b(n) \gets \ell cp(b(v), b(par) \concat e_{par})$ \Comment{Node between $v$ and its parent}
		\State $n'\gets \textsc{DHT-Search}(b(n))$
		\If{$n\neq nil \textbf{ and } n\notin\MsdSet $}
		\State $n' \gets \textsc{Linearize}(v)$\Comment{Present $v$}
		\Else
		\State $n \gets \textsc{MultiLinearize}(\{v,\, par,\, par \concat e_{par}\})$ \Comment{Create $n$ in \DHT}
		\State \textsc{DHT-Insert}($n$)
		\EndIf
		\Else
		\If{$v,\,par\in\PatSet \textbf{ and } e_{par}=p_{-}(v)$}
		\If{$\textsc{MsdMissing}(v,\,par)$}\label{algorithm:line:missing_msd_node_insertion}
		\State $m \gets \text{new Msd node}$\Comment{Insert Msd node}
		\State $b(m) \gets \textsc{MsdLabel}(par, v)$
		\State $m \gets \textsc{MultiLinearize}(\{v,\,par\})$
		\State $m'\gets\textsc{DHT-Search}(b(m))$
		\If{$m'=nil \textbf{ or } (m'\text{ has different edges than } m$\\\hspace{5.2cm}$\textbf{and } m'\in\MsdSet)$}
		\State \textsc{DHT-Insert}($m$)\Comment{Overwrite with correct Msd node}
		\EndIf
		\EndIf
		\EndIf
		\EndIf
		\EndIf
		\EndIf
		\EndProcedure
	\end{algorithmic}
\end{algorithm}
\noindent
\cref{algorithm:Trie-Node_parent_edge_check} checks if the parent edge of the \HPT\ node $v$ is correct.
If necessary, the edge is corrected, a root is created or new nodes are inserted.

\newpage
\begin{table}[H]
	\setlength{\tabcolsep}{5pt}
	\renewcommand{\arraystretch}{1.5}
	\begin{tabular}{|l|l|}
		\hline
		\textbf{Line(s)} & \textbf{Affected Lemmas}                                                                                                                                                       \\\hline
		$4-7$            & \cref{lemma:root_node_exists}                                                                                                                                                  \\\hline
		$8-9$            & \cref{lemma:existing_parent_nodes_established}                                                                                                                                 \\\hline
		$12-13$          & \cref{lemma:deleted_not_rebuilt_nodes_references_vanish}, \cref{lemma:root_node_exists}, \cref{lemma:existing_parent_nodes_established}, \cref{lemma:No_Patricia_Point_To_Msd} \\\hline
		$16-24$          & \cref{lemma:inner_pat_nodes_without_rigths_leave}, \cref{lemma:existing_parent_nodes_established}, \cref{lemma:non_leaving_nodes_reached_from_root}                            \\\hline
		$26-34$          & \cref{lemma:incorrect_msd_nodes_removed}, \cref{lemma:all_missing_msd_nodes_created}                                                                                           \\\hline
	\end{tabular}
	\noindent\vspace{0.4cm}
	\caption{Lemmas affected by \cref{algorithm:Trie-Node_parent_edge_check}}
\end{table}

\begin{algorithm}[H]
	\caption{Checking child edges of $v\in\TrieSet$}
	\label{algorithm:Trie-Node_child_edges_check}
	\begin{algorithmic}[1]
		\Procedure{CheckChildEdgeInfo}{$v\in\TrieSet$}\Comment{Check each child}
		\If{$v\in\PatSet$}
		\State \textsc{CheckChild}($v,\, \textsc{DHT-Search}(b(v) \concat p_{0}(v)), 0$)
		\State \textsc{CheckChild}($v,\, \textsc{DHT-Search}(b(v) \concat p_{1}(v)), 1$)
		\EndIf
		\EndProcedure
		\Procedure{CheckChild}{$v,\, c \in \TrieSet$, $x\in\{0,1\}$}
		\If{$c = nil$\textbf{ or }$c\in\MsdSet$}\label{algorithm:line:non_existing_child_node_edge_deleted}
		\State $p_{x}(v) = nil$\Comment{Child non-existing or Msd node}
		\EndIf
		\EndProcedure
	\end{algorithmic}
\end{algorithm}
\noindent
\cref{algorithm:Trie-Node_child_edges_check} checks if the child edges of $v$ point to valid Patricia nodes.\\\\

\begin{table}[H]
	\setlength{\tabcolsep}{5pt}
	\renewcommand{\arraystretch}{1.5}
	\begin{tabular}{|l|l|}
		\hline
		\textbf{Line(s)} & \textbf{Affected Lemmas}                                                                                                                          \\\hline
		$6-7$            & \cref{lemma:deleted_not_rebuilt_nodes_references_vanish}, \cref{lemma:last_nodes_with_empty_child_deleted}, \cref{lemma:No_Patricia_Point_To_Msd} \\\hline
	\end{tabular}
	\noindent\vspace{0.4cm}
	\caption{Lemmas affected by \cref{algorithm:Trie-Node_child_edges_check}}
\end{table}

\newpage
\begin{algorithm}[H]
	\caption{Checking whether $v\in\TrieSet$ should exist}
	\label{algorithm:Trie-Node_validity_check}
	\begin{algorithmic}[1]
		\Procedure{CheckValidity}{$v\in\TrieSet$}
		\If{$v\in\MsdSet$}
		\If{$p_{-}(v)\neq nil \textbf{ and } \textsc{MsdChildEdge}(v)\neq nil$}
		\State $p \gets \textsc{DHT-Search}(\textsc{Parent}(v))$
		\State $c \gets \textsc{DHT-Search}(b(v)\concat \textsc{MsdChildEdge}(v))$
		\If{$\textbf{not }(\textsc{Bidirectional}(p,\,c)$ \textbf{ and } $p,\,c\in\PatSet$\\\hspace{2.6cm}\textbf{and } $b(v)=\textsc{MsdLabel}(b(p),b(c)))$}
			\State $\textbf{delete }v$\label{algorithm:line:msd_node_not_patricia_between_bidirected_correct_form_deleted}\Comment{Incorrect Msd node}
			\EndIf
			\Else
			\State $\textbf{delete }v$\label{algorithm:line:msd_node_no_child_parent_deleted}\Comment{No parent edge or not one child edge}
			\EndIf
			\Else
			\If{$key(v) = nil \textbf{ and } \textsc{Children}(v)<2 \text{ and } v \neq root$}\label{algorithm:line:patricia_node_no_key_no_two_children_no_root_deleted}
			\State \textbf{delete} $v$ \Comment{Unnecessary Patricia node}
		\EndIf
		\EndIf
		\EndProcedure
	\end{algorithmic}
\end{algorithm}
\noindent
\cref{algorithm:Trie-Node_validity_check} checks if $v$ should exist at all.
In case that $v$ lacks child edges, it is removed.\\\\

\begin{table}[H]
	\setlength{\tabcolsep}{5pt}
	\renewcommand{\arraystretch}{1.5}
	\begin{tabular}{|l|l|}
		\hline
		\textbf{Line(s)} & \textbf{Affected Lemmas}                                                                                                                                                                                           \\\hline
		$3-8$            & \cref{lemma:no_more_deletions_happen}, \cref{lemma:incorrect_msd_nodes_removed}, \cref{lemma:all_missing_msd_nodes_created}                                                                                        \\\hline
		$9-10$           & \cref{lemma:empty_subtree_nodes_vanish}, \cref{lemma:no_more_deletions_happen}, \cref{lemma:incorrect_msd_nodes_removed}, \cref{lemma:all_missing_msd_nodes_created}                                               \\\hline
		$12-13$          & \cref{lemma:empty_subtree_nodes_vanish}, \cref{lemma:inner_pat_nodes_without_rigths_leave}, \cref{lemma:last_nodes_with_empty_child_deleted}, \cref{lemma:no_more_deletions_happen}, \cref{lemma:root_node_exists} \\\hline
	\end{tabular}
	\noindent\vspace{0.4cm}
	\caption{Lemmas affected by \cref{algorithm:Trie-Node_validity_check}}
\end{table}

\newpage
\begin{algorithm}[H]
	\caption{Checking if $v\in\TrieSet$ has valid key$_{2}$ information}
	\label{algorithm:Trie-Node_key2_check}
	\begin{algorithmic}[1]
		\Procedure{CheckKey$2$Info}{$v\in\TrieSet$}
		\If{$v\in\PatSet$}
		\If{$\textsc{Children}(v)=2$}\Comment{Inner node}
		\If{$key_{2}(v) \neq nil$}
		\State $k \gets \textsc{DHT-Search}(key_{2}(v))$
		\If{$k = nil \textbf{ or } k \in\MsdSet \textbf{ or } \textsc{Children}(k) > 0 \textbf{ or } b(v) \sqsubset r(k)$}\label{algorithm:line:delete_key2_value_when_not_leaf_pointer}
		\State $key_{2}(v) \gets nil$\Comment{$v$ should not store $k$}
		\ElsIf{$r(k)=nil$\textbf{ or }$r(k) \sqsubset b(v)$}\label{algorithm:line:check_if_r_should_set}
		\State $r(k) \gets b(v)$\Comment{$v$ references $k$}
		\EndIf
		\EndIf
		\If{$key_{2}(v) = nil$}\label{algorithm:line:capture_key2_from_above}
		\State $k \gets \text{ next Patricia node above }v\text{ such that }b(v)\sqsubset key_2(k)$
		\State $key_{2}(v) \gets key_{2}(k)$
		\EndIf
		\ElsIf{$\textsc{Children}(v)=0$}\Comment{leaf node}\label{algorithm:line:leaf_node_check_for_r_values}
		\If{$r(v) \neq nil$}
		\State $k \gets \textsc{DHT-Search}(r(v))$
		\If{$k = nil \textbf{ or } \textsc{Children}(k)<2 \textbf{ or } key_{2}(k) \neq b(v)$}\label{algorithm:line:r_value_connected_to_not_suitable_node}
		\State $r(v) \gets nil$\Comment{$k$ not suitable}
		\ElsIf{$key_{2}(k) = nil$}\label{algorithm:line:repair_key2_if_reference_points}
		\State $key_{2}(k) \gets b(v)$\Comment{Repair reference}
		\EndIf
		\Else
		\State $k \gets w\in\PatSet \text{ with } \textsc{Children}(w) = 2,\text{ } key_{2}(w)\in\{nil,b(v)\}$\\\hspace*{2.5cm}$\text{ above }v\text{ with }|b(w)|\text{ maximal}$\label{algorithm:line:find_good_key2_node}
		\State $r(v) \gets k$\Comment{Find node for $r(v)$}
		\State $key_{2}(k) \gets b(v)$
		\EndIf
		\EndIf
		\EndIf
		\EndProcedure
	\end{algorithmic}
\end{algorithm}
\noindent
\cref{algorithm:Trie-Node_key2_check} checks if the key$_2$/$r$ information stored at $v$ is valid.
If necessary, a new key$_2$ or $r$ value is obtained.\\\\

\begin{table}[H]
	\setlength{\tabcolsep}{5pt}
	\renewcommand{\arraystretch}{1.5}
	\begin{tabular}{|l|l|}
		\hline
		\textbf{Line(s)} & \textbf{Affected Lemmas}                                                                             \\\hline
		$6-7$            & \cref{lemma:strange_references_key2_disappear}, \cref{lemma:inner_nodes_with_two_children_save_key2} \\\hline
		$8-9$            & \cref{lemma:not_eligible_nodes_drop_key2}, \cref{lemma:inner_nodes_with_two_children_save_key2}      \\\hline
		$10-12$          & \cref{lemma:strange_references_key2_disappear}, \cref{lemma:inner_nodes_with_two_children_save_key2} \\\hline
		$16-17$          & \cref{lemma:strange_references_key2_disappear}, \cref{lemma:r_values_corrected}                      \\\hline
		$18-19$          & \cref{lemma:not_eligible_nodes_drop_key2}, \cref{lemma:strange_references_key2_disappear}            \\\hline
		$20-24$          & \cref{lemma:not_eligible_nodes_drop_key2}, \cref{lemma:inner_nodes_with_two_children_save_key2}      \\\hline
	\end{tabular}
	\noindent\vspace{0.4cm}
	\caption{Lemmas affected by \cref{algorithm:Trie-Node_key2_check}}
\end{table}

\newpage
\noindent
For better readability, we denote by $v \leftarrow \textsc{Linearize}(u)$ a call of $\textsc{Linearize}(v,\,u)$ at the \DHT\ node responsible for $v$.
Similar, $v \leftarrow \textsc{MultiLinearize}(U)$ translates to $\textsc{MultiLinearize}(v,\, U)$ at the \DHT\ node responsible for $v$.

\begin{algorithm}[H]
	\caption{Linearization methods at \DHT\ node $n$}
	\label{algorithm:Linearization}
	\begin{algorithmic}[1]
		\Procedure{LinearizeTimeout}{$v\in\TrieSet$}
		\If{$v\in\PatSet$}
		\State $\textsc{Parent}(v) \gets \textsc{Linearize}(b(v))$
		\State $b(v) \concat p_{0}(v) \gets \textsc{Linearize}(v)$
		\State $b(v) \concat p_{1}(v) \gets \textsc{Linearize}(v)$
		\EndIf
		\EndProcedure
		\Procedure{MultiLinearize}{$v\in\TrieSet ,\, U\in \TrieSet^{*}$}
		\ForEach{$u\in U$}
		\State $\textsc{Linearize}(v,\,u)$
		\EndFor
		\EndProcedure
		\Procedure{Linearize}{$v,\,u\in \TrieSet$}\Comment{Present $u$ to $v$}
		\If{$ \textbf{not }(v\in\MsdSet\textbf{ or } u\in\MsdSet)$ \textbf{and} $v\neq u$}\label{algorithm:line:no_edge_when_presenting_another_msd}
		\If{$\vert \ell cp(b(u), b(v)) \vert < \vert b(v)\vert$}
		\If{$p_{-}(v)=nil \textbf{ and } b(u)\sqsubseteq b(v)$} \Comment{$u$ above $v$}
		\State $p_{-}(v) \gets x \text{ [where }b(u) \concat x = b(v)]$
		\Else
		\If{$\textsc{Parent}(v) \neq b(u)$}
		\State $\textsc{Parent}(v) \gets \textsc{Linearize}(u)$
		\EndIf
		\If{$\textsc{Parent}(v) \sqsubset b(u) \textbf{ and } b(u)\sqsubset b(v)$} \Comment{$u$ in between}
		\State $p_{-}(v) \gets x \text{ [where }b(u) \concat x = b(v)]$
		\EndIf
		\EndIf
		\ElsIf{$b(v) \sqsubset b(u)$} \Comment{$v$ above $u$}
		\State $x \gets b(u)\text{ at position }\vert b(v) \vert +1$
		\State $c \gets b(v) \concat p_{x}(v)$ \Comment{Respective child edge}
		\If{$p_{x}(v)=nil$}
		\State $p_{x}(v) \gets y \text{ [where } b(v) \concat y = b(u)]$
		\Else
		\If{$b(c) \sqsubset b(u)$} \Comment{$c$ above $u$}
		\State $c \gets \textsc{Linearize}(u)$
		\ElsIf{$b(u) \sqsubset b(c)$} \Comment{$c$ below $u$}
		\State $p_{x}(v) \gets y \text{ [where } b(v) \concat y = b(u)]$
		\State $c \gets \textsc{Linearize}(u)$
		\Else \Comment{Common parent for $c,\, u$ needed}
		\State $u \gets \textsc{Linearize}(v)$
		\EndIf
		\EndIf
		\EndIf
		\EndIf
		\EndProcedure
	\end{algorithmic}
\end{algorithm}
\noindent
\cref{algorithm:Linearization} denotes the Linearization procedure of our protocol.
A Patricia node $v$ tries to maintain connections to the closest Patricia nodes of the form $w$ where either $b(v)\sqsubset b(w)$ or $b(w)\sqsubset b(v)$.
Presentations of other Patricia nodes are delegated and Msd nodes are not presented nor do they delegate messages.

\newpage

\section{Correctness Proof}\label{appendix:correctness_proof}

In the following, a complete correctness proof follows.
Assuming that the underlying \DHT\ is in a legal state, we show that our protocol converges the system to a state in which the hashed Patricia Trie is in a legal state with respect to \cref{definition:Correctness_of_HPT} and our modification as presented in \cref{section:HashedPatriciaTrie}.
Further, once such a legal state for the \HPT\ is reached, closure holds and the \HPT\ stays legal in consecutive states.
For this, we show a sequence of lemmas, each representing properties that the system converges to and that are preserved.

In the proofs, we often define \emph{potential functions}.
We denote by a potential function a function depending on the current system state of which the value decreases over time, i.e., in following states, and is not increased.
Often, our potential functions depend on a set of elements.
For convenience, we assume that as soon as the corresponding set is empty, the function is defined to be zero.
In some cases, there may be messages in the system that present a non-existent \HPT\ node.
These messages may temporarily increase our potential functions.
However, on the long term, they are removed from the system.
Therefore, we often implicitly assume in our arguments that this has happened and ignore temporary increases of potential functions.
We aim at proving \cref{theorem:HPTCorrect} with respect to the complete definition of a legal state of the \HPT\ below.

\setcounter{PotentialIndex}{0}

\CorrectnessTheorem*

\begin{restatable}[Legal state of the \HPT]{definition}{DefinitionCorrectHPT}
	\label{definition:Correctness_of_HPT}
	We say a \HPT\ is \text{in a legal state}, if the following holds.
	Only Patricia nodes given by the following conditions exist.
	Every $k\in\KeySet$ is stored in a Patricia node $v$ with $b(v)=k$.
	For every pair of keys $k_1,k_2$, there exists an inner Patricia node $v$ with $b(v)=\ell cp(k_1,k_2)$.
	Only Msd nodes given by the following exist:
	Between every pair $u,v$ of closest Patricia nodes, there is exactly one Msd node $m$ (if needed) with $b(m)$ as given by \cref{definition:msd_calculation}, if $b(u)\sqsubset b(m) \sqsubset b(v)$.
	$m$ has a parent edge to $u$ and exactly one child edge to $v$.
	In general, no \HPT\ node $v$ exists with $key(v)=nil$ and there is no $k\in\KeySet$ such that $b(v)\sqsubseteq k$.
	Also, every \HPT\ node is stored at the \DHT\ node responsible for it.

	The edges of the \HPT\ need to be correct as defined in the following.
	For every Branch Set $S$, there are bidirectional valid edges between every pair of closest Patricia nodes of $S$.
	Every leaf node $v\in\PatSet$ stores $r(v)=w$ for a $w\in\PatSet$ with two children, such that $key_{2}(w)=v$ and $b(w)\sqsubseteq b(v)$.
	No other nodes than those affected by this statement store a value in $r$ or key$_2$.
\end{restatable}

\lemmaMainMiscRepair*

\begin{proof}
	The proof of the lemma is given by the proofs of \cref{lemma:keys_saved_at_right_node}, \cref{lemma:any_node_saved_correct_in_dht}, \cref{corollary:keys_searchable} and \cref{lemma:deleted_not_rebuilt_nodes_references_vanish}.
\end{proof}

\stepcounter{PotentialIndex}
\begin{restatable}{lem}{lemmaKeysSavedAtRightNode}
	\label{lemma:keys_saved_at_right_node}
	Let $\phi_{\thePotentialIndex}=\vert\{k\in\KeySet\,\vert\, k \text{ is not stored in a } v\in\PatSet,\text{ or } b(v)\neq k \}\vert$ be the number of keys which are stored by a Patricia node of wrong label.
	$\phi_{\thePotentialIndex}$ is a potential function.
\end{restatable}

\begin{proof}
	Consider any key $k$ not stored in a Patricia node.
	In finite time, $k$ is detected and stored as a Patricia node by the respective \DHT\ node $n$ (see \cref{algorithm:DHT_node_protocol} \cref{algoithm:line:save_Key_In_Patricia_Node}).
	Let $k \in \KeySet$ be a key stored at $v \in \PatSet$ with $b(v) \neq k$.
	In finite time our algorithm performs a check of $v$ detecting that $b(v) \neq k$.
	As a result, $k$ is reinserted in a node $w \in \PatSet$ with $b(w) = k$ such that $\phi_{\thePotentialIndex}$ is reduced (see \cref{algorithm:Trie-Node_local_check} \cref{algorithm:line:wrong_saved_key}).
	Our protocol does not insert a key $k$ at node $v$ where $b(v)\neq k$.
	Hence, $\phi_{\thePotentialIndex}$ cannot increase.
\end{proof}

We call a \HPT\ node \emph{searchable} if it is stored at the \DHT\ node responsible for it and can be retrieved by \textsc{DHT-Search}.

\begin{restatable}{lem}{lemmaAnyNodeSavedCorrectInDht}
	\label{lemma:any_node_saved_correct_in_dht}
	In finite time every \HPT\ node is searchable.
\end{restatable}

\begin{proof}
	Consider any node $v \in \TrieSet$ stored at a \DHT\ node $n$ which is not responsible for $v$.
	In finite time, our protocol checks if $v$ is stored correctly in the \DHT\ (see \cref{algorithm:DHT_node_protocol} \cref{algorithm:line:delegation_wrong_saved_info}).
	As this is not the case, $v$ is moved to the corresponding \DHT\ node and becomes searchable.
\end{proof}

\noindent
From \cref{lemma:keys_saved_at_right_node} and \cref{lemma:any_node_saved_correct_in_dht} it directly follows:

\begin{restatable}{corollary}{corollaryKeysSearchable}
	\label{corollary:keys_searchable}
	In finite time, every $k \in \KeySet$ is searchable.
\end{restatable}

\begin{restatable}{lem}{lemmaDeletedNotRebuiltNodesReferencesVanish}
	\label{lemma:deleted_not_rebuilt_nodes_references_vanish}
	Consider any $v\in \TrieSet$ that is deleted.
	As long as $v$ is not reconstructed, in finite time it holds:
	\begin{enumerate}
		\item There is no presentation message for $b(v)$.
		\item There is no edge pointing towards $b(v)$ in the system.
	\end{enumerate}
\end{restatable}

\begin{proof}
	$ $\\
	a) Since $v$ does not exist and is not reconstructed, no new presentation messages for $b(v)$ are generated.
	Since the \DHT\ is in a legal state, every message containing $b(v)$ eventually arrives at its destination \DHT\ node.
	Furthermore, a delegation in the \HPT\ only happens finitely often, because there is only a finite number of nodes and a presentation is not delegated in circles.
	This directly follows from the applied Linearization procedure.
	For details on Linearization, we refer to \cite{Onus:2007:LLS:2791188.2791198}.
	Eventually, a presentation is processed and perhaps a new edge towards $v$ in the \HPT\ is created.
	In any case the number of messages is reduced.
	\smallskip\\
	b) From a) it follows that in finite time no presentation messages concerning $b(v)$ are present.
	Let $w\in\TrieSet$ be a node pointing to $v$.
	In finite time, $w$ checks the corresponding edge and determines that $v$ is non-existing (see \cref{algorithm:Trie-Node_parent_edge_check} \cref{algorithm:line:non_existing_parent_node_edge_deleted} and \cref{algorithm:Trie-Node_child_edges_check} \cref{algorithm:line:non_existing_child_node_edge_deleted}).
	The edge is deleted and will never be rebuilt, because there are no more presentation messages concerning $v$.
\end{proof}

\subsubsection*{Phase \rom{1} -- Deletion of Patricia nodes}

\lemmaMainEmptySubtreesNotNeededNodesRemoved*

\begin{proof}
	The lemma is correct by \cref{lemma:initially_wrong_reference_information_vanishes}, \cref{lemma:empty_subtree_nodes_vanish} and \cref{lemma:inner_pat_nodes_without_rigths_leave}.
\end{proof}

\stepcounter{PotentialIndex}
\begin{restatable}{lem}{lemmaInitiallyWrongReferenceInformationVanishes}
	\label{lemma:initially_wrong_reference_information_vanishes}
	Consider the number of initially wrong reference information
	\begin{align*}
		\phi_{\thePotentialIndex} = & \vert \{ v \in \TrieSet \,\vert\, \nexists\, x\in\{0,1\}^{*} \text{ such that } x \concat p_{-}(v) = b(v) \} \vert \\
		                            & + \vert \{ v \in \TrieSet \,\vert\, p_{0}(v) \text{ does not comply } 0(0\vert 1)^{*} \} \vert                     \\
		                            & + \vert \{ v \in \TrieSet \,\vert\, p_{1}(v) \text{ does not comply } 1(0\vert 1)^{*} \} \vert                     \\
		                            & + \vert \{ v \in \PatSet \,\vert\, r(v) \neq nil \text{ and } r(v) \not\sqsubseteq b(v) \} \vert                   \\
		                            & + \vert \{ v \in \PatSet \,\vert\, key_{2} \neq nil \text{ and } b(v) \not\sqsubseteq key_{2}(v) \} \vert .
	\end{align*}
	$\phi_{\thePotentialIndex}$ is a potential function.
\end{restatable}

\begin{proof}
	All the information types gathered in $\phi_{\thePotentialIndex}$ are checked locally and without acquiring other nodes.
	When some of the types appear, the corresponding edges are set to $nil$, reducing $\phi_{\thePotentialIndex}$ (see \cref{algorithm:Trie-Node_local_check} \cref{algorithm:line:parent_edge_no_suffix}, \cref{algorithm:line:p0_wrong_form}, \cref{algorithm:line:p1_wrong_form}, \cref{algorithm:line:key2_wrong_form}, \cref{algorithm:line:r_wrong_form}).
	The protocol never stores information in a way as defined by $\phi_{\thePotentialIndex}$, such that closure holds.
\end{proof}

\stepcounter{PotentialIndex}
\begin{restatable}{lem}{lemmaEmptySubtreeNodesVanish}
	\label{lemma:empty_subtree_nodes_vanish}
	Let $\phi_{\thePotentialIndex}=\max_{ v \in \mathcal{U} } \vert b(v) \vert$, where:
	\begin{align*}
		\mathcal{U}=\{v\in \TrieSet \,\vert\, & key(v)=nil \text{ and } \nexists\, w\in\PatSet                      \\
		                                      & \text{with } key(w) \neq nil \text{ and } b(v) \sqsubseteq b(w) \}.
	\end{align*}
	$\mathcal{U}$ is the set of all \HPT\ nodes in subtrees that can safely be deleted.
	$\phi_{\thePotentialIndex}$ is a potential function.
\end{restatable}

\begin{proof}
	Let $\mathcal{T}=\{ v\in \mathcal{U} \,\vert\, \vert b(v) \vert = \phi_{\thePotentialIndex} \}$.
	We show that $\vert \mathcal{T} \vert$ is reduced and never increased as long as $\phi_{\thePotentialIndex}$ stays the same.
	When $\vert \mathcal{T} \vert$ reaches zero, $\phi_{\thePotentialIndex}$ is reduced and $\mathcal{T}$ changes.
	We show that $\phi_{\thePotentialIndex}$ does not increase and follow that $\phi_{\thePotentialIndex}$ is a potential function.

	Consider a $v\in \mathcal{T}$.
	If $v \in \MsdSet$, then $v$ has no existing child.
	Else, the child of $v$ would have a longer label than $v$, such that $v\notin \mathcal{T}$.
	$v$ will be deleted in finite time by the protocol (see \cref{algorithm:Trie-Node_validity_check} \cref{algorithm:line:msd_node_no_child_parent_deleted}).
	In the other case, let $v \in \PatSet$.
	Again, $v$ has no existing child since $\vert b(v) \vert =\phi_{\thePotentialIndex}$ and any child would be a member of $\mathcal{U}$ with a longer label.
	According to \cref{lemma:deleted_not_rebuilt_nodes_references_vanish}, any child edge of $v$ will be deleted in finite time.
	Afterwards, $v$ deletes itself since $key(v)=nil$ (see \cref{algorithm:Trie-Node_validity_check} \cref{algorithm:line:patricia_node_no_key_no_two_children_no_root_deleted}).
	In any case, $\vert \mathcal{T} \vert$ is reduced.
	Also, $\vert \mathcal{T} \vert$ and $\phi_\thePotentialIndex$ will never be increased, because a node $u$ can only be created by a node $w$ where $b(u) \sqsubset b(w)$.
	The existence of $w$ would contradict the maximality of $\mathcal{T}$ concerning $\phi_\thePotentialIndex$.
\end{proof}

\stepcounter{PotentialIndex}
\begin{restatable}{lem}{lemmaInnerPatNodesWithoutRigthsLeave}
	\label{lemma:inner_pat_nodes_without_rigths_leave}
	Let $\phi_{\thePotentialIndex}=\vert\, \mathcal{F} \,\vert$, where
	\begin{align*}
		\mathcal{F} = \{ v\in\PatSet \,\vert\, & key(v) = nil \text{ and } \nexists\, u,\, w \in \PatSet \\
		                                       & \text{with } key(u), \, key(w) \neq nil                 \\
		                                       & \text{and } b(v) = \ell cp(b(u), b(w)) \}.
	\end{align*}
	$\mathcal{F}$ is the set of all remaining Patricia nodes that are not needed in the \HPT.
	$\phi_{\thePotentialIndex}$ is a potential function.
\end{restatable}

\begin{proof}
	Consider a node $v \in \mathcal{F}$.
	We have that $v \notin \mathcal{U}$ as $\mathcal{U}$ is empty in finite time due to \cref{lemma:empty_subtree_nodes_vanish}.
	Hence, $v$ has at least one non-existing child node $w$ which will never be created.
	This holds, because by definition there is no $u \in \PatSet$ with $key(u)\neq nil$ and $b(w) \sqsubseteq b(u)$.
	According to \cref{lemma:deleted_not_rebuilt_nodes_references_vanish}, the edge to $w$ will be deleted at $v$ in finite time such that $v$ deletes itself resulting in a reduction of $\phi_{\thePotentialIndex}$ (see \cref{algorithm:Trie-Node_validity_check} \cref{algorithm:line:patricia_node_no_key_no_two_children_no_root_deleted}).
	A new Patricia node $v$ is only created if there are two Patricia nodes $u,\, w$ with $key(u),\, key(w)\neq nil$ and $b(v)=\ell cp(b(u),b(w))$ (see \cref{algorithm:Trie-Node_parent_edge_check}).
	This results in $v \notin \mathcal{F}$ and therefore, $\phi_{\thePotentialIndex}$ is not increased.
\end{proof}

\lemmaMainDeletionStops*

\begin{proof}
	The lemma holds due to the correctness of \cref{lemma:last_nodes_with_empty_child_deleted} and \cref{lemma:no_more_deletions_happen}.
\end{proof}

\stepcounter{PotentialIndex}
\begin{restatable}{lem}{lemmaLastNodesWithEmptyChildDeleted}
	\label{lemma:last_nodes_with_empty_child_deleted}
	Let $\mathcal{N}$ be the set of nodes, defined by their labels, which do not exist in the \HPT.
	Let $\phi_{\thePotentialIndex} = \max_{ v \in \mathcal{D}} \vert b(v) \vert$, where
	\begin{align*}
		\mathcal{D} = & \{ v\in\MsdSet\cup\mathcal{N} \,\vert\, \text{there is a } v\in\PatSet \text{ with } b(w)\sqsubseteq b(v) \text{ that points to } v \} \\
		              & \cup \{ m \in \MsgSet \,\vert\, b(m) = b(v), \, v \in\MsdSet\cup\mathcal{N}\cup\mathcal{L} \} \cup \mathcal{L}                         \\
		\mathcal{L} = & \{ v\in\PatSet \,\vert\, key(v)=nil \text{ and } p_{0}(v)=nil \text{ or } p_{1}(v)=nil \}
	\end{align*}
	$\mathcal{L}$ is the set of Patricia nodes that are unnecessary in their local perspective.
	$\mathcal{D}$ is the set of nodes and messages that may result in a deletion of a Patricia node.
	$\phi_{\thePotentialIndex}$ is a potential function.
\end{restatable}

\begin{proof}
	Similar to the proof of \cref{lemma:empty_subtree_nodes_vanish}, we define $\mathcal{T} = \{ e \in \mathcal{D} \,\vert\, \vert b(e) \vert = \phi_{\thePotentialIndex} \}$ to be the set of all elements in $\mathcal{D}$ of longest label.
	We will show that as long as $\phi_{\thePotentialIndex}$ stays the same, $\vert \mathcal{T} \vert$ is reduced and never increases.
	When $\vert \mathcal{T} \vert$ reaches zero, $\phi_{\thePotentialIndex}$ is decreased and $\mathcal{T}$ contains a new set of elements.
	We also show that $\phi_{\thePotentialIndex}$ is never increased and follow that $\phi_{\thePotentialIndex}$ is a potential function.
	For $e\in S$, we distinguish the following cases:
	\begin{enumerate}
		\item $e \in \MsdSet\cup\mathcal{N}$ and there is a finite set $\mathcal{A}$ of Patricia nodes which have a child edge pointing to $e$.
		      Consider any Patricia node $v$ out of this set.
		      Either $e$ is overwritten by a Patricia node, $v$ is presented a Patricia node (see \cref{algorithm:Linearization}) or $v's$ respective child edge $p_{x}(v)$ is checked in finite time.
		      In the first case, $e$ is no longer in $\mathcal{T}$ such that $|\mathcal{T}|$ reduced.
		      In the other two cases, $p_{x}(v)$ is set to $nil$ (see \cref{algorithm:Trie-Node_child_edges_check} \cref{algorithm:line:non_existing_child_node_edge_deleted}).
		      As $|b(v)|<|b(e)|$ and $|\mathcal{A}|$ is finite, this implies a reduction of $|\mathcal{T}|$ in finite time.
		\item $e \in \MsgSet$ with $b(e) = b(v)$ for some node $v$ which does not exist or is an Msd node.
		      If a new Patricia node $v$ is created, we argue below that this node is not in $\mathcal{T}$.
		      Further, $e$ is no longer in $\mathcal{T}$ in this case such that $|\mathcal{T}|$ reduced.
		      Else, due to \cref{lemma:deleted_not_rebuilt_nodes_references_vanish}, $e$ vanishes in finite time and may only result in one element in $\mathcal{T}$ for which case a) applies.
		\item $e \in \MsgSet$ with $b(e)=b(v)$ and $v$ is a Patricia node in $\mathcal{L}$.
		      Then $v \in \mathcal{T}$ and case d) applies to $v$.
		\item $e \in \mathcal{L}$.
		      Either $e$ is presented an existing Patricia node (see \cref{algorithm:Linearization}), or $e$ does not contribute to the system in its local view.
		      If $e$ is presented an existing Patricia node, $|\mathcal{T}|$ reduces and all messages presenting $e$ are no longer in $\mathcal{T}$.
		      If $e$ does not contribute to the system in its local view, it will be deleted in finite time (see \cref{algorithm:Trie-Node_validity_check} \cref{algorithm:line:patricia_node_no_key_no_two_children_no_root_deleted}) and $|\mathcal{T}|$ reduces.
	\end{enumerate}
	In all cases $\vert \mathcal{T} \vert$ is eventually reduced.
	It is left to show that $\phi_\thePotentialIndex$ is not increased and $\vert \mathcal{T}\vert$ is not increased.
	Both could be increased if a Patricia node changes its child edge, if a node is inserted or if a Patricia node is deleted.

	Assume a node $v \in \PatSet$ changes its child edge such that $\phi_\thePotentialIndex$ or $\vert\mathcal{T}\vert$ increases.
	If $v$ deletes its child edge, the edge pointed to a node $n$ which was either non-existing or an Msd node.
	As the deletion increased $\phi_\thePotentialIndex$ or $|\mathcal{T}|$, $|b(n)|$ was greater than $\phi_\thePotentialIndex$ which is a contradiction because $n\in\mathcal{D}$.
	If $v$ changes its edge, let $m$ be the responsible message presenting $\vert b(w) \vert \geq \phi_{\thePotentialIndex}$.
	It holds $m\in\mathcal{D}$.
	As $m$ was processed when $v's$ child edge was created, $\vert\mathcal{T}\vert$ did not increase.
	An increase of $\phi_\thePotentialIndex$ poses a contradiction, because $m$ was already in $\mathcal{D}$.

	Consider the insertion of a node.
	Assume the creation of an Msd node $m$ with $\vert b(m) \vert \geq \phi_{\thePotentialIndex}$.
	Msd nodes are not presented and no Patricia node is overwritten by an Msd node.
	Hence, the existence of a node $v\in\PatSet$ above $m$ that has a child edge to $m$ implies that $v$ previously had a child edge to a non-existing node $n$ with $|b(n)|=|b(m)|$.
	If $\phi_\thePotentialIndex$ increased due to $m$, this poses a contradiction.
	Also, $|\mathcal{T}|$ was not increased because $n$ does not count towards $|\mathcal{T}|$ after insertion of $m$.

	Assume the creation of a Patricia node $v$ with $\vert b(v) \vert \geq \phi_{\thePotentialIndex}$.
	$v$ has been created, because there are two Patricia nodes $u,\,w$ with $b(v)=\ell cp(b(u), b(w))$.
	By definition $u,\, w \notin \mathcal{T}$ since their labels are longer than $\vert b(v) \vert \geq \phi_{\thePotentialIndex}$.
	Assume $w$ created $v$.
	Then $w$ is a Patricia node.
	$v$ could only be in $\mathcal{D}$ if there was a Patricia node above $v$ having a child edge pointing to a node $m\in\MsdSet\cup\mathcal{N}$.
	$v$ has in this case both $m$ and $w$ as initial children such that $|b(m)|>|b(v)|$.
	Since $m\in\mathcal{D}$ this implies that $|b(v)|<\phi_\thePotentialIndex$ resulting in a contradiction.

	Consider the deletion of a node $v$.
	The deletion of an Msd node does not influence $\phi_\thePotentialIndex$ or $|\mathcal{T}|$, therefore assume $v\in\PatSet$.
	If $v\in\mathcal{D}$, then $|b(v)|\leq\phi_\thePotentialIndex$ and the deletion of $v$ does not increase $\phi_\thePotentialIndex$ or $|\mathcal{T}|$.
	Assuming that $\phi_\thePotentialIndex$ or $|\mathcal{T}|$ increased, $v\notin\mathcal{D}$.
	Thus, $v$ either stores a key or has two existing nodes as children.
	This implies that $v$ is not deleted posing a contradiction (see \cref{algorithm:Trie-Node_validity_check} \cref{algorithm:line:patricia_node_no_key_no_two_children_no_root_deleted}).

	Hence, $\phi_{\thePotentialIndex}$ and $|\mathcal{T}|$ do not increase and decreases over time which proves the lemma.
\end{proof}

\begin{restatable}{lem}{lemmaNoMoreDeletionsHappen}
	\label{lemma:no_more_deletions_happen}
	In finite time no further Patricia nodes will be deleted during stabilization and all messages concerning non-existing Patricia nodes vanished.
\end{restatable}

\begin{proof}
	Due to \cref{lemma:last_nodes_with_empty_child_deleted}, in finite time every Patricia node falls in one of the following categories:
	\begin{enumerate}
		\item Leaf nodes storing a key.
		\item Inner Patricia nodes storing a key.
		\item Inner Patricia nodes with two existing Patricia nodes as child nodes.
	\end{enumerate}
	For each of these categories it holds, that the conditions leading to deletion are not fulfilled (deletions happen in \cref{algorithm:Trie-Node_local_check} \cref{algorithm:line:wrong_saved_key} and \cref{algorithm:Trie-Node_validity_check} \cref{algorithm:line:msd_node_not_patricia_between_bidirected_correct_form_deleted}, \cref{algorithm:line:msd_node_no_child_parent_deleted}, \cref{algorithm:line:patricia_node_no_key_no_two_children_no_root_deleted}).
	Furthermore, no Patricia nodes which do not fall into one of the categories will be created by the protocol, because we proved closure of the state of \cref{lemma:last_nodes_with_empty_child_deleted}.
	\cref{lemma:deleted_not_rebuilt_nodes_references_vanish} assures that all messages concerning non-existing nodes vanish in finite time.
\end{proof}

\subsubsection*{Phase \rom{2} -- Reconstruction}

\lemmaMainReconstructionRootCorrectPointing*

\begin{proof}
	We prove the lemma by proving \cref{lemma:root_node_exists}, \cref{lemma:existing_parent_nodes_established}, \cref{corollary:path_from_every_node_to_root}, \cref{lemma:No_Patricia_Point_To_Msd}, \cref{lemma:non_leaving_nodes_reached_from_root} and \cref{lemma:missing_patricia_nodes_created}.
\end{proof}

\begin{restatable}{lem}{lemmaRootNodeExists}
	\label{lemma:root_node_exists}
	In finite time, a root node exists.
\end{restatable}

\begin{proof}
	Assume that there is no root node.
	Our first observation is, that there must be a node $v\in\PatSet$ with shortest label.
	Hence, either $p_{-}(v) = nil$ or $v$ has a parent node $w$ which is non-existent or an Msd node.
	In the latter case, the edge $p_{-}(v)$ is checked and set to $nil$ in finite time (see \cref{algorithm:Trie-Node_parent_edge_check} \cref{algorithm:line:non_existing_parent_node_edge_deleted}).
	If $p_{-}(v)=nil$, $v$ checks in finite time for a parent using \textsc{BinaryPrefixSearch}.
	The resulting node is $nil$, such that $v$ restores the root node (see \cref{algorithm:Trie-Node_parent_edge_check} \cref{algorithm:line:root_creation}).
	Furthermore, the root node is never deleted (see \cref{algorithm:Trie-Node_validity_check} \cref{algorithm:line:patricia_node_no_key_no_two_children_no_root_deleted}).
\end{proof}

\stepcounter{PotentialIndex}
\begin{restatable}{lem}{lemmaExistingParentNodesEstablished}
	\label{lemma:existing_parent_nodes_established}
	Let $\phi_{\thePotentialIndex} = \max_{k\in\KeySet} \vert k \vert - \min_{v\in\mathcal{P}} \vert b(v) \vert$, where:
	\begin{align*}
		\mathcal{P} = \{ v \in \PatSet \,\vert\, & p_{-}(v) = nil \text{ or } \nexists\, u \in \PatSet, \, b(u) \concat p_{-}(v) = b(v) \}.
	\end{align*}
	$\mathcal{P}$ is the set of all Patricia nodes which do not have an existing Patricia node as parent node.
	$\phi_{\thePotentialIndex}$ is a potential function, i.e., in finite time every Patricia node has a parent edge to an existing Patricia node.
\end{restatable}

\begin{proof}
	First we observe that $\max_{k\in \KeySet} \vert k \vert \geq \vert b(v) \vert$ for all $v\in\PatSet$ such that $\phi_{\thePotentialIndex} \geq 0$.
	We define $\mathcal{T}=\{v\in\mathcal{P}\,\vert\, \max_{k\in \KeySet} \vert k \vert - \vert b(v) \vert = \phi_{\thePotentialIndex} \}$.
	We will show that $\vert \mathcal{T}\vert$ decreases and never increases as long as $\phi_{\thePotentialIndex}$ stays the same.
	When $\vert \mathcal{T}\vert$ reaches zero, $\phi_{\thePotentialIndex}$ is reduced.
	$\phi_{\thePotentialIndex}$ is never increased and it follows that $\phi_{\thePotentialIndex}$ is a potential function.

	Consider a Patricia node $v \in \mathcal{T}$.
	If the node $u$ with $b(u) \concat p_{-}(v) = b(v)$ does not exist or is an Msd node, the value of $p_{-}(v)$ is deleted in finite time by the protocol (see \cref{algorithm:Trie-Node_parent_edge_check} \cref{algorithm:line:non_existing_parent_node_edge_deleted}).
	If $p_{-}(v) = nil$, $v$ will be presented a Patricia node $w$ with $b(w)\sqsubset b(v)$ either due to a presentation message or due to a check of $v$ (see \cref{algorithm:Linearization} and \cref{algorithm:Trie-Node_parent_edge_check} \cref{algorithm:line:new_parent_due_to_binary_prefix_search}), because at least the root node exists as stated in \cref{lemma:root_node_exists}.
	In any case, $\vert \mathcal{T}\vert$ has reduced.

	Furthermore, $\vert \mathcal{T}\vert$ and $\phi_\thePotentialIndex$ do not increase.
	Consider any Patricia node $v$ with $\max_{k\in \KeySet} \vert k \vert - \vert b(v) \vert \geq \phi_{\thePotentialIndex}$.
	If the edge $p_{-}(v)$ changes, $v$ received a presentation message from an existing Patricia node $w$ with $b(w) \sqsubset b(v)$.
	This holds as non-existing nodes and Msd nodes are not presented and, due to \cref{lemma:no_more_deletions_happen}, no Patricia node is deleted any more.

	If a new Patricia node $v$ is created it node is initialized with edges provided by a Patricia node $w$ with $b(v) \sqsubset b(w)$ (see \cref{algorithm:Trie-Node_parent_edge_check} \cref{algorithm:line:parent_should_be_patricia_node_with_two_children}).
	$v \notin \mathcal{P}$ holds, because $w$ only creates a node if it has an existing Patricia node as parent such that $w\notin\mathcal{P}$ holds.

	Hence, $\phi_{\thePotentialIndex}$ does not increase and is reduced over time.
\end{proof}

\begin{restatable}{corollary}{corollaryPathFromEveryNodeToRoot}
	\label{corollary:path_from_every_node_to_root}
	In finite time, there exists a path from every node to the root.
\end{restatable}

The corollary follows from \cref{lemma:existing_parent_nodes_established} and \cref{lemma:no_more_deletions_happen}.

\begin{restatable}{lem}{lemmaNoPatriciaPointToMsd}
	\label{lemma:No_Patricia_Point_To_Msd}
	In finite time, no Patricia node has an edge pointing to an Msd node.
\end{restatable}

\begin{proof}
	Due to \cref{lemma:no_more_deletions_happen}, no Patricia node $v$ not storing a key which has a child edge pointing to an Msd node exists.
	Else, $v$ would delete itself which is a contradiction.
	\cref{lemma:existing_parent_nodes_established} assures that no Patricia node has a parent edge pointing to an Msd node.
	Furthermore, every Patricia node storing a key which points to an Msd node will delete the respective edge in finite time (see \cref{algorithm:Trie-Node_parent_edge_check} \cref{algorithm:line:non_existing_parent_node_edge_deleted} and \cref{algorithm:Trie-Node_child_edges_check} \cref{algorithm:line:non_existing_child_node_edge_deleted}).
	Patricia nodes do not create edges pointing to Msd nodes (see \cref{algorithm:Linearization}).
	When a Patricia node $v$ is created, it is inserted between $u\in\PatSet$ and $w\in\PatSet$ with initial edges to $u$ and $w$.
	Additionally, $v$ has another child node $p$ which was formerly a child of $u$ and is thus no Msd node.
	Therefore, newly inserted Patricia nodes do not point to Msd nodes as well.
\end{proof}

\begin{restatable}{lem}{lemmaNonLeavingNodesReachedFromRoot}
	\label{lemma:non_leaving_nodes_reached_from_root}
	In finite time, every Patricia node can be reached over a path starting at the root.
\end{restatable}

\begin{proof}
	We already know that every Patricia node has a path to the root node in finite time.
	The correctness of \cref{lemma:non_leaving_nodes_reached_from_root} follows as we use a technique called Linearization \cite{Onus:2007:LLS:2791188.2791198}.
	In \cite{Onus:2007:LLS:2791188.2791198}, the authors show that their technique creates in finite time a sorted list.
	We apply Linearization for every path consisting of Patricia nodes from the root to a Patricia node storing a key.
	If Patricia nodes are missing at positions where two Branch Sets collide, such nodes are inserted (see \cref{algorithm:Trie-Node_parent_edge_check} \cref{algorithm:line:parent_should_be_patricia_node_with_two_children}).
	As stated in \cref{lemma:No_Patricia_Point_To_Msd}, no Patricia node points to an Msd node, so we can ignore Msd nodes when considering Linearization.
	It is left to show that insertion of Patricia nodes does not harm the Linearization process.
	Consider the case when a Patricia node $v$ is inserted between $u\in\PatSet$ and $w\in\PatSet$.
	Initially, $v$ has valid edges to $u$ and $w$ such that the connectivity between them is not destroyed.
\end{proof}

\stepcounter{PotentialIndex}
\begin{restatable}{lem}{lemmaMissingPatriciaNodesCreated}
	\label{lemma:missing_patricia_nodes_created}
	Let $\phi_{\thePotentialIndex} = \vert \mathcal{M} \vert$ where:
	\begin{align*}
		\mathcal{M} = \{ v \in \PatSet \,\vert\, & \nexists\, v \text{ and } \exists\, u,\,w \in \PatSet \\
		                                         & \text{with } key(u),\, key(w) \neq nil                \\
		                                         & \text{and } b(v) = \ell cp(b(u), b(w)) \}.
	\end{align*}
	$\mathcal{M}$ denotes the set of all Patricia nodes that are needed in the \HPT\ but currently non-existing.
	$\phi_{\thePotentialIndex}$ is a potential function.
\end{restatable}

\begin{proof}
	Assume that there are paths from every node to the root and from the root to every node as stated in \cref{corollary:path_from_every_node_to_root} and \cref{lemma:non_leaving_nodes_reached_from_root}.
	Additionally, let $\phi_{\thePotentialIndex} > 0$.
	Consider a node $v \in \mathcal{M}$.
	For the nodes $u,\, w\in \PatSet$ with $key(u), \, key(w) \neq nil$ and $b(v) = \ell cp(b(u), b(w))$ it holds that they can be reached over a path starting from the root and the root can be reached over a path starting at $u$ or $w$.
	This implies either the existence of $v$ or an edge at a Patricia node pointing to a non-existent node or an Msd node.
	In all cases we have a contradiction.

	According to \cref{lemma:no_more_deletions_happen}, no deletions happen any more.
	Hence, $\phi_{\thePotentialIndex}$ is not increased.
\end{proof}

\lemmaMainLinearizationWorksMsdNodesRemovedAndInserted*

\begin{proof}
	The lemma holds by the correctness of \cref{lemma:Every_Branch_Set_Bidirectional_Edges}, \cref{lemma:incorrect_msd_nodes_removed}, \cref{lemma:all_missing_msd_nodes_created} and \cref{corollary:structure_rebuilt}.
\end{proof}

\begin{restatable}{lem}{lemmaEveryBranchSetBidirectionalEdges}
	\label{lemma:Every_Branch_Set_Bidirectional_Edges}
	In finite time for every Branch Set $S$ it holds:
	Between every pair of closest Patricia nodes $u,\,w\in S$ there is a bidirectional edge.
\end{restatable}

\begin{proof}
	Consider any Branch Set $S$.
	\cref{lemma:non_leaving_nodes_reached_from_root} and \cref{corollary:path_from_every_node_to_root} assure that a weak connectivity is given for the set of Patricia nodes in $S$.
	Observe that due to \cref{lemma:No_Patricia_Point_To_Msd}, Msd nodes do not influence the stabilization of $S$.
	Furthermore, \cref{lemma:no_more_deletions_happen} and \cref{lemma:missing_patricia_nodes_created} assure that no more Patricia nodes are deleted or created such that $S$ does not change.
	All Patricia nodes in $S$ perform a Linearization procedure derived from the proposition in \cite{Onus:2007:LLS:2791188.2791198}.
	Therefore, in finite time, a sorted list is created for all nodes in $S$.
	This implies the lemma.
\end{proof}

\stepcounter{PotentialIndex}
\begin{restatable}{lem}{lemmaIncorrectMsdNodesRemoved}
	\label{lemma:incorrect_msd_nodes_removed}
	We call an Msd node $m$ \emph{incorrect}, if $m$ does not have a parent edge to $u\in\PatSet$, $m$ does not have a child edge to $w\in\PatSet$, $u$ and $w$ are not connected by a bidirectional edge, or $b(m)$ has incorrect length (see \cref{definition:msd_calculation}).
	Let $\mathcal{I}$ be the set of incorrect Msd nodes in the system.
	$\phi_\thePotentialIndex = |\mathcal{I}|$ is a potential function.
\end{restatable}

\begin{proof}
	We observe that the set of incorrect Msd nodes does not increase after every Branch Set consists of bidirectional edges.
	This holds as new Msd nodes are only created if there is a bidirectional edge between two Patricia nodes $u,\, w$ (see \cref{algorithm:Trie-Node_parent_edge_check} \cref{algorithm:line:missing_msd_node_insertion}).
	If such an edge exists, then an inserted Msd node $m$ is correct by the correctness of the creation.
	If $m$ became incorrect, this would mean a Patricia node was inserted between $u$ and $v$.
	This is a contradiction to the assumption that the Branch Set of $u$ and $w$ only had bidirectional edges when $m$ was created.

	Consider an Msd node $m\in\mathcal{I}$.
	$m$ is checked in finite time and is determined to be incorrect (see \cref{algorithm:Trie-Node_validity_check} \cref{algorithm:line:msd_node_not_patricia_between_bidirected_correct_form_deleted}, \cref{algorithm:line:msd_node_no_child_parent_deleted}) and deleted.
	Thus, $|\mathcal{I}|$ is reduced.

	As the set of incorrect Msd nodes is finite, it follows that all incorrect Msd nodes are removed in finite time.
\end{proof}

\begin{restatable}{definition}{definitionMissingMsdNodes}
	\label{definition:missing_msd_nodes}
	Consider $u,\, w\in \PatSet$ with $b(u) \sqsubset b(w)$ and $\nexists\, v\in \PatSet$ such that $b(u) \sqsubset b(v) \sqsubset b(w)$.
	We call an Msd node $m$ of correct form (as given by \cref{definition:msd_calculation}) between $u$ and $w$ \emph{missing} if $m$ does not exist.
\end{restatable}

\begin{restatable}{lem}{lemmaAllMissingMsdNodesCreated}
	\label{lemma:all_missing_msd_nodes_created}
	In finite time all missing Msd nodes are created.
\end{restatable}

\begin{proof}
	Consider two Patricia nodes $u,\, w\in \PatSet$ with $b(u) \sqsubset b(w)$ where an Msd node $m$ is missing.
	This means that there is no $v \in \PatSet$ between $u$ and $w$.
	Due to \cref{lemma:Every_Branch_Set_Bidirectional_Edges}, there is a bidirectional edge between $u$ and $w$.
	When $w$ is checked, it inserts $m$ in finite time (see \cref{algorithm:Trie-Node_parent_edge_check} \cref{algorithm:line:missing_msd_node_insertion}).
	According to \cref{lemma:no_more_deletions_happen}, no more deletions happen.
	In addition, $m$ has a parent edge to $u\in\PatSet$ and a child edge to $w\in\PatSet$, the edge between $u$ and $w$ is bidirectional and $b(m)$ has correct length.
	Therefore, $m$ is not deleted (see \cref{algorithm:Trie-Node_validity_check} \cref{algorithm:line:msd_node_not_patricia_between_bidirected_correct_form_deleted}, \cref{algorithm:line:msd_node_no_child_parent_deleted}), the number of missing Msd nodes does not increase and \cref{lemma:all_missing_msd_nodes_created} follows.
\end{proof}

\begin{restatable}{corollary}{corollaryStructureRebuilt}
	\label{corollary:structure_rebuilt}
	In finite time the structure of the \HPT\ is rebuilt.
\end{restatable}

\cref{corollary:structure_rebuilt} follows directly as Patricia nodes are connected by bidirectional edges, no incorrect Msd nodes or unnecessary Patricia nodes exist, and all missing Msd nodes have been created.

\subsubsection*{Phase \rom{3} -- Consistency}

\lemmaMainKeyTwoRValueReferencesOnlyEligible*

\begin{proof}
	The lemma is correct by the correctness of \cref{lemma:not_eligible_nodes_drop_key2} and \cref{lemma:strange_references_key2_disappear}.
\end{proof}

\begin{restatable}{lem}{lemmaNotEligibleNodesDropKeyTwo}
	\label{lemma:not_eligible_nodes_drop_key2}
	In finite time, only key$_2$ nodes store a key$_2$ and only leaf nodes store an $r$ value.
\end{restatable}

\begin{proof}
	Consider a Patricia node $v$ which is no key$_2$ node but stores a key$_2$.
	In finite time, $v$ is checked and sets $key_2(v)=nil$ (see \cref{algorithm:Trie-Node_local_check} \cref{algorithm:line:key2_wrong_form}).
	No $v\in\PatSet$ which is not a key$_2$ node starts to store a key$_{2}$ due to our protocol (see \cref{algorithm:Trie-Node_key2_check} \cref{algorithm:line:repair_key2_if_reference_points}, \cref{algorithm:line:find_good_key2_node}).

	Consider a Patricia node $v$ which is not a leaf node but stores an $r$ value.
	$v$ sets $r(v)=nil$ in finite time (see \cref{algorithm:Trie-Node_local_check} \cref{algorithm:line:r_wrong_form}, \cref{algorithm:line:r_value_stored_at_inner_node}).
	No inner node stores an $r$ value based on the protocol (see \cref{algorithm:Trie-Node_key2_check} \cref{algorithm:line:check_if_r_should_set}, \cref{algorithm:line:leaf_node_check_for_r_values}).
\end{proof}

\begin{restatable}{lem}{lemmaStrangeReferencesKeyTwoDisappear}
	\label{lemma:strange_references_key2_disappear}
	In finite time, every key$_{2}$ value stored at $v\in\PatSet$ refers to a leaf node below $v$ and every $r$ value stored at $w\in\PatSet$ refers to a key$_2$ node above $w$.
\end{restatable}

\begin{proof}
	According to \cref{lemma:not_eligible_nodes_drop_key2}, in finite time only inner nodes with two children store key$_{2}$ values and only leaf nodes store $r$ values.
	When a node $v\in \PatSet$ detects that $key_{2}(v)$ is not the label of a leaf node below it, such a reference is deleted (see \cref{algorithm:Trie-Node_local_check} \cref{algorithm:line:key2_wrong_form} and \cref{algorithm:Trie-Node_key2_check} \cref{algorithm:line:delete_key2_value_when_not_leaf_pointer}).
	Also, when a leaf node $v$ detects that $r(v)$ is not the label of a key$_2$ node above it, $r(v)$ is set to $nil$ (see \cref{algorithm:Trie-Node_local_check} \cref{algorithm:line:r_wrong_form} and \cref{algorithm:Trie-Node_key2_check} \cref{algorithm:line:r_value_connected_to_not_suitable_node}).
	Nodes which are not key$_2$ nodes are not considered when propagating keys and only labels of leaf nodes are propagated using parent edges only (see \cref{algorithm:Trie-Node_key2_check} \cref{algorithm:line:leaf_node_check_for_r_values}, \cref{algorithm:line:repair_key2_if_reference_points}).
	Also, key$_2$ nodes only acquire key$_2$ values that point to leaf nodes below them (see \cref{algorithm:Trie-Node_key2_check} \cref{algorithm:line:capture_key2_from_above}).
	Therefore, the state described in the lemma is reached.
\end{proof}

\lemmaMainKeyTwoRValueReferencesMadeConsistent*

\begin{proof}
	The lemma follows from \cref{lemma:inner_nodes_with_two_children_save_key2} and \cref{lemma:r_values_corrected}.
\end{proof}

\stepcounter{PotentialIndex}
\begin{restatable}{lem}{lemmaInnerNodesWithTwoChildrenSaveKeyTwo}
	\label{lemma:inner_nodes_with_two_children_save_key2}
	Let $\phi_{\thePotentialIndex}=\max_{v\in \mathcal{K}} \vert b(v) \vert$, where:
	\begin{align*}
		\mathcal{K} = & \{ v\in\PatSet \,\vert\, \textsc{Children}(v)=2 \text{ or } v=root \}                                                              \\
		              & \cap \{ v\in\PatSet \,\vert\, key_{2}(v)=nil \text{ or } \exists\, w\in\PatSet:\, key_{2}(w)=key_{2}(v),\, b(v) \sqsubset b(w) \}.
	\end{align*}
	$\mathcal{K}$ is the set of all key$_2$ nodes not storing a key$_2$ that stays.
	$\phi_{\thePotentialIndex}$ is a potential function.
\end{restatable}

\begin{proof}
	Let $\mathcal{T}=\{ v \in \mathcal{K} \,\vert\, \vert b(v) \vert = \phi_{\thePotentialIndex}\}$.
	We will show that $\vert \mathcal{T}\vert$ is reduced and never increases as long as $\phi_{\thePotentialIndex}$ stays the same.
	Further, $\phi_{\thePotentialIndex}$ does not increase.
	When $\vert \mathcal{T}\vert$ reaches zero, $\phi_{\thePotentialIndex}$ is reduced.
	As a result, the lemma follows.

	Consider $v\in \mathcal{T}$.
	If $key_{2}(v) \neq nil$, then there has to be a node $w$ with $key_{2}(w)=key_{2}(v)$ and $b(v) \sqsubset b(w)$.
	Since $v\in \mathcal{T}$, exactly one such node $w$ exists.
	In finite time, $w$ checks $u:=key_{2}(w)$ and sets $r(u)=w$ (see \cref{algorithm:Trie-Node_key2_check} \cref{algorithm:line:check_if_r_should_set}).
	Afterwards, $v$ detects that $b(v) \sqsubset r(u)$ and sets $key_{2}(v) = nil$ (see \cref{algorithm:Trie-Node_key2_check} \cref{algorithm:line:delete_key2_value_when_not_leaf_pointer}).
	So, assume $key_2(v)=nil$.
	Due to \cref{fact:enough_key2_spaces_for_leaves}, there is a leaf node $u$ below $v$, such that either (a) no inner node stores $b(u)$ as a key$_{2}$, or (b) some key$_2$ node $w$ with $b(w)\sqsubset b(v)$ has $key_{2}(w)=b(u)$.
	Let $w$ be maximal considering $\vert b(w) \vert$.
	In (a) it holds that $b(u)$ is in finite time presented to $v$, because $b(v)\sqsubset b(u)$ and every node between $v$ and $u$ already stores a key$_{2}$ (see \cref{algorithm:Trie-Node_key2_check} \cref{algorithm:line:find_good_key2_node}).
	In (b) $b(w) \sqsubset b(v) \sqsubset b(u)$ holds and a message from $v$ reaches $w$ in finite time such that $v$ sets $b(u)$ as key$_2$ (see \cref{algorithm:Trie-Node_key2_check} \cref{algorithm:line:capture_key2_from_above}).
	In any case, $\vert \mathcal{T}\vert$ reduces.

	No node $v$ with $\vert b(v) \vert > \phi_{\thePotentialIndex}$ deletes its $key_{2}(v)$, because this only happens if there is a node $w$ with $b(v) \sqsubset b(w)$ and $key_{2}(w)=key_{2}(v)$ which is not the case by definition.
	Therefore, $\phi_{\thePotentialIndex}$ is never increased.
	Since no node with a legitimate key$_{2}$ drops this value, no deletions happen (\cref{lemma:no_more_deletions_happen}), and there is no node $v\in \mathcal{K}$ with $\vert b(v) \vert > \phi_{\thePotentialIndex}$, $\vert \mathcal{T}\vert$ is never increased until it reaches zero.
	Hence, the lemma follows.
\end{proof}

\stepcounter{PotentialIndex}
\begin{restatable}{lem}{lemmaRValuesCorrected}
	\label{lemma:r_values_corrected}
	Let $\phi_{\thePotentialIndex}=\vert \mathcal{R} \vert$, where:
	\begin{align*}
		\mathcal{R} = \{v\in \PatSet \,\vert\, \textsc{Children}(v) = 0, \text{ } r(v) = nil \text{ or } key_{2}(r(v)) \neq b(v) \}\setminus\{root\}.
	\end{align*}
	$\mathcal{R}$ denotes the set of leaf nodes without an $r$ value that stays stable.
	$\phi_{\thePotentialIndex}$ is a potential function.
\end{restatable}

\begin{proof}
	Assume that we reached a state where every inner node has a stable key$_{2}$ (\cref{lemma:inner_nodes_with_two_children_save_key2}).
	Consider a node $v\in \mathcal{R}$.
	If $r(v) \neq nil$ and $key_{2}(r(v)) \neq b(v)$, then $v$ sets $r(v)=nil$ in finite time (see \cref{algorithm:Trie-Node_key2_check} \cref{algorithm:line:r_value_connected_to_not_suitable_node}).
	If $r(v) = nil$, then there must be a node $w$ with $key_{2}(w)=b(v)$.
	This node will correct $r(v)$ in finite time to $r(v)=w$ leading to a reduction of $\phi_{\thePotentialIndex}$ (see \cref{algorithm:Trie-Node_key2_check} \cref{algorithm:line:r_value_connected_to_not_suitable_node}).
	Furthermore, no node $v\in \PatSet$ with $\textsc{Children}(v)=0$ and $key_{2}(r(v))=b(v)$ will delete $r(v)$ and the key$_{2}$ values stay fixed according to the closure of \cref{lemma:inner_nodes_with_two_children_save_key2}.
	Hence, $\phi_{\thePotentialIndex}$ is never increases.
	\cref{fact:enough_key2_spaces_for_leaves} assures that $\mathcal{R}$ is empty in finite time.
\end{proof}

\newpage
\section{Overhead Proof}\label{appendix:overhead_proof}

\theoremConstantAmountOfMessagesPerTimeout*

\begin{proof}
	On a call of \textsc{Timeout}, exactly one \HPT\ node is checked per \DHT\ node.
	This \HPT\ node $v$ has at most three edges provided by $p_{-}(v)$, $p_{0}(v)$ and $p_{1}(v)$, and one reference provided by either $key_{2}(v)$ or $r(v)$.
	In total, no more than four nodes may be (partially) acquired using the \DHT.
	Furthermore, Linearization presents $v$ to at most three other nodes.
	When the \HPT\ is in a legal state, no reinsertion, no presenting of $b(v)$ as a key$_2$ value and no searching of other key$_2$ values is done.
	Hence, we have $\Theta(1)$ created hash table accesses and messages at a \DHT\ node.
\end{proof}

\theoremStorageOverhead*

\begin{proof}
	In \cite{HashedPatriciaTrie} it was mentioned that the \HPT\ needs $\Theta (\bitsSym)$ memory space if $\bitsSym=\sum_{k\in \KeySet} \vert k \vert$ is the number of bits needed to store all keys.
	The modifications we made do not change the number of Patricia nodes and Msd nodes.
	Our protocol rebuilds the structure of the \HPT\ and deletes every unnecessary Patricia node as well as every incorrect Msd node (see \cref{lemma:inner_pat_nodes_without_rigths_leave} and \cref{lemma:incorrect_msd_nodes_removed}).
	Thus, the rebuilt \HPT\ has an asymptotically optimal memory demand of \MemoryBound\ bits.
\end{proof}

\end{document}